\documentclass[twoside,11pt,final]{entics} 

\usepackage{thmtools}
\usepackage{enticsmacro}
\usepackage{graphicx}
\usepackage{quiver}
\usepackage{booktabs}
\usepackage{makecell}
\usepackage{hyperref}
\usepackage{color}
\usepackage{stmaryrd}
\usepackage{amsmath}
\usepackage[all]{xy}

 \newcommand{\Cat}{\mathbf{Cat}}

\newcommand{\profto}{\mathrel{\ooalign{$\longrightarrow$\cr\hfil\!$\mid$\hfil\cr}}}

\newcommand{\collapse}{\mathit{collapse}}
\newcommand{\Sat}{\mathit{sat}}
\newcommand{\op}{{\rm op}}

\newcommand{\cov}{\mathbin{{{\mathrel-\joinrel\subset}}}}
\newcommand{\longcov}[1]{{\stackrel{#1}{\mathrel-\joinrel\relbar\joinrel\subset}}}

\renewcommand{\max}{{\it top}}

\newcommand{\imc}{\rightarrowtriangle}
\newcommand{\setdif}{\setminus}

\newcommand{\sig}{\sigma}

\newcommand{\fsubseteq}{\subseteq_{\rm fin}}

\newcommand{\CC}{{\rm C\!\!C}}
\newcommand{\cc}{\mathrm{\,c\!\!\!\!c}\,}
\providecommand{\pol}{}\renewcommand{\pol}{{\mathrm{pol}}}

\newcommand{\Strat}{{\mathbf{Strat}}}

\newcommand{\F}{{\mathcal F}}
\newcommand{\G}{{\mathcal G}}

\newcommand{\ESS}{{\mathbf {EvSym}}}
\newcommand{\ESSP}{\ESS_\mathrm{p}}
\newcommand{\ES}{{\mathbf {Ev}}}

\def\Con{{\rm Con}}

\newcommand{\Fam}{{\cal F}}

\newcommand{\arr}[1]{{{\stackrel{#1}{\longrightarrow}}}}
\newcommand{\id}{{\rm id}}

\newcommand{\arrow}{\rightarrow}
\newcommand{\set}[2]{{\{  #1\  | \  #2 \} }}
\newcommand{\setof}[1]{{\{ #1 \} }}

\newcommand{\Set}{\mathbf{Set}}
\newcommand{\pSet}{\mathbf{pSet}}

\newcommand{\eqdef}{\coloneqq}

\newcommand{\iso}{\cong}

\newcommand{\ie}{{\it i.e.}}

\usepackage{amsfonts}
\usepackage{tikz-cd}
\usepackage{amssymb}
\usepackage{mathrsfs}
\usepackage{mathtools}
\usepackage{todonotes}
\newcommand{\confsym}[1]{\mathscr C(#1)}

\newcommand{\strongequiv}{\simeq_{\mathrm{s}}}
\newcommand{\weakequiv}{\simeq}

\newcommand{\sat}[1]{\mathrm{sat}(#1)}
\newcommand{\E}{\mathcal{E}}
\newcommand{\D}{\mathcal{D}}

\newcommand{\longto}{\longrightarrow}
\newcommand{\xto}{\xrightarrow}
\newcommand{\xiso}[1]{\mathrel{\mathrlap{\smash{\xto[\smash{\raisebox{1.3mm}{$\scriptstyle\sim$}}]{#1}}}\hphantom{\xto{#1}}}}
\newcommand\toiso{\xto{\smash{\raisebox{-.5mm}{$\scriptstyle\sim$}}}}
\newcommand\lefttoiso{\xleftarrow{\smash{\raisebox{-.5mm}{$\scriptstyle\sim$}}}}

\newcommand{\tosym}[1]{\xtosym{}{#1}}
\newcommand{\xtosym}[2]{
  \mathrel{
    \mathrlap{\smash{\xto[\smash{\raisebox{0mm}{\scriptsize $\widetilde{#2}$}}]{\
          #1\ }}}
    \hphantom{\ \xto{#1}\ }}
}
\newcommand{\pto}{\rightharpoonup}
\newcommand{\sym}[1]{\widetilde{#1}}
\newcommand{\Setoid}{\mathbf{Setoid}}
\newcommand{\SetoidProf}{\Setoid\text{-}\mathbf{Prof}}
\newcommand{\SetoidFib}{\mathbf{SetoidFib}}
\newcommand{\SetoidSFib}{\mathbf{SetoidSFib}}
\newcommand{\calC}{\mathcal{C}}
\newcommand{\calD}{\mathcal{D}}
\newcommand{\calE}{\mathcal{E}}
\newcommand{\calA}{\mathcal{A}}
\newcommand{\calB}{\mathcal{B}}

\newcommand{\fibre}[2]{#1^{-1}({#2})}
\newcommand{\essfibre}[2]{{#1_\sim^{-1}}({#2})}
\newcommand{\fibres}{\mathbf{fibres}}
\newcommand{\essfibres}{\mathbf{ess\text{-}fibres}}
\newcommand{\inter}{\circledast}
\newcommand{\comp}{\odot}
\DeclarePairedDelimiter\coll\|\|
\newcommand{\config}[1]{{\mathscr C}(#1)}

\newcommand{\posincl}{\mathrel{\lhook\joinrel\xrightarrow{\raisebox{-0.35ex}{$\scriptscriptstyle\boxplus$}}}}
\newcommand{\negincl}{\mathrel{\lhook\joinrel\xrightarrow{\raisebox{-0.35ex}{$\scriptscriptstyle\boxminus$}}}}
\newcommand{\negrev}{\mathrel{\xleftarrow{\raisebox{-0.35ex}{$\scriptscriptstyle\boxminus$}}\joinrel\rhook}}

\newcommand{\bplus}{{\scriptscriptstyle\boxplus}}
\newcommand{\bminus}{{\scriptscriptstyle\boxminus}}

\usetikzlibrary{calc,fit,matrix,decorations.markings,decorations.pathreplacing,arrows,cd,positioning,shapes.misc}
\usetikzlibrary{decorations.pathmorphing}
\usetikzlibrary{backgrounds}
\tikzstyle{game-causality}=[dotted, thick]
\tikzstyle{strat-causality}=[->, thick,  -open triangle 60]
\tikzstyle{conflict}=[decorate, decoration={snake,amplitude=.3mm,segment length=2mm},-]

\tikzset{
  posnode/.style args={#1}{
  fill=blue!25, draw, thick,  inner sep=1pt, minimum size=10pt,    
    label={[font=\scriptsize, inner sep=2pt]270:$#1$},
path picture={
    \draw[thick]
      (path picture bounding box.north) -- (path picture bounding box.south)
      (path picture bounding box.west) -- (path picture bounding box.east);
  }}}
\tikzset{
  negnode/.style={
  fill=red!25, draw, thick, inner sep=1pt, minimum size=10pt,
path picture={
    \draw[thick]
      (path picture bounding box.west) -- (path picture bounding box.east);
  }}}
\tikzset{
  negnodel/.style args={#1}{
  fill=red!25, draw, thick, inner sep=1pt, minimum size=10pt,
    label={[font=\scriptsize, inner sep=2pt]270:$#1$},
path picture={
    \draw[thick]
      (path picture bounding box.west) -- (path picture bounding box.east);
  }}}
  \tikzset{
frame/.style={
  draw,
  rounded corners=6pt,
  thick,
  fill=none,
  inner sep=10pt
}}

\def\conf{MFPS 2026} 	%
\volume{6}			%
\def\volu{6}			%
\def\papno{15}			%
\def\mydoi{19550}
\def\lastname{Paquet, Winskel} %
\def\runtitle{Concurrent strategies as Street fibrations} %
\def\copynames{H. Paquet, G. Winskel}    %
\begin{document}
\begin{frontmatter}
  \title{Concurrent Strategies as Street Fibrations}
  \author{Hugo Paquet\thanksref{a}\thanksref{myemail}}	%
   \author{Glynn Winskel\thanksref{b}\thanksref{coemail}}		%
   \address[a]{Inria and \'Ecole Normale Sup\'erieure\\				%
    Paris, France}
   \thanks[myemail]{Email: \href{mailto:hugo.paquet@inria.fr} {\texttt{\normalshape
        hugo.paquet@inria.fr}}}
  \address[b]{Queen Mary University of London\\
    United Kingdom}
  \thanks[coemail]{Email:  \href{mailto:g.winskel@qmul.ac.uk} {\texttt{\normalshape
        g.winskel@qmul.ac.uk}}}
\begin{abstract}
Concurrent games and strategies based on event structures are an established and expressive model of interactive computation. In this paper we revisit the problem of equipping this model with symmetry. Symmetry is essential for many applications, but it is also a source of mathematical complexity, not least because strategies with symmetry can be compared up to several distinct notions of equivalence. 

Our first contribution is a complete characterization of two classes of strategies for which composition admits an identity. For these two classes the identity laws hold, respectively, up to a `weak' and a `strong' notion of equivalence. The main result is that the weak class of strategies corresponds precisely to those inducing Street fibrations over the game. We also show how the established `thin' approach to symmetry also fits in this general framework. 

A second contribution is a formal connection between games and setoid-valued profunctors. This follows a long tradition of relating games and relational models, but here additionally includes a very general model of symmetries. Our characterization of strategies as Street fibrations makes the connection to profunctors significantly easier to establish. 
\end{abstract}
\end{frontmatter}

\section{Introduction}

Event structures \cite{NPW,evstrs} are a simple mathematical model for concurrency: a set of computational events, constrained by causality and conflict relations, describes the possible behaviours of a concurrent process.

More recently event structures have also provided the foundation for an expressive theory of games and strategies with concurrent features \cite{lics11}. In that context, the `events' are occurrences of moves in two-player games, and a strategy is a concurrent process describing the behaviour of one player as it interacts with the opponent over the course of the game. These concurrent games and strategies form an expressive model of program interaction, relying on the mathematics of event structures. This model already admits several extensions \cite{quantstrat,ictac2015,aurore,mfps2018} and is intimately connected to other instances of concurrency in game semantics \cite{FP,Asgames,Congames,murawski,asynchronousTemplate}.

Applications of concurrent games to programming language semantics (e.g., \cite{lics15,lics18}) have required a generalization of the model in which event structures are equipped with `symmetry', an equivalence relation that expresses when executions are essentially the same \cite{ESS}. This generalization is essential (e.g.~to support monads and comonads) but it greatly complicates the mathematical theory: with symmetry, strategies require a new composition operation and new identity morphisms, and new notions of equivalence weaker than isomorphism \cite{lics14}. 

In this paper we revisit symmetry in concurrent games with hindsight, making connections to established categorical structures for managing symmetries. We construct a new model of games and strategies with symmetry, arguing that it is the `maximally general' model, which embeds all previous constructions. Furthermore, we give a new characterization of `weak' strategies as Street fibrations (\S\ref{subsec:intro-characterization}), and establish a formal connection to pseudo-profunctors (\S\ref{subsec:intro-profunctors}). 

\subsection{The problem of characterizing strategies}
\label{subsec:intro-characterization}

Concurrent strategies interact and compose (see \S\ref{sec:games-and-strategies}), but it is a technical challenge to show that this composition obeys the expected laws of composition. In particular, the existence of an identity strategy (the `copycat') satisfying the left and right unit laws is not automatic, and relies on certain conditions being imposed on strategies. These conditions characterize strategies within a broader class of strategy-like structures (the `pre-strategies') for which composition is not necessarily unital \cite{lics11}. 

The problem of characterizing strategies is more difficult for games with symmetry. In \cite{lics14} a set of conditions is shown to be sufficient for a strategy with symmetry to satisfy the unit laws up to a notion of equivalence of strategies called \emph{strong equivalence} (\S\ref{subsec:pre-strategies}). This result falls short of a characterization: these sufficient conditions are in fact not all necessary. We will remedy this to give a complete characterization of \emph{strong strategies} (those satisfying the laws up to strong equivalence). 

In this paper we also consider a weaker notion of equivalence of strategies, perhaps more natural from a categorical point of view, called \emph{weak equivalence}. Strong equivalence implies weak equivalence, and therefore the class of \emph{weak strategies} (those satisfying the unit laws up to weak equivalence) necessarily includes the strong strategies. Our main result (Theorem~\ref{thm:characterization}) is that weak strategies also admit a characterization, namely as forming Street fibrations over a particular category induced by the game. (Street fibrations are a weakening of the usual definition of fibrations, due to Ross Street~\cite{street1980fibrations}, with the advantage of being stable under equivalence of categories. We give precise definitions in \S\ref{sec:fibrations}.) This result extends an existing characterization of strategies without symmetry: strategies are discrete fibrations \cite{fossacs13}. 

\subsection{Strategies and profunctors}
\label{subsec:intro-profunctors}

It is well known that games and strategies share many features with categories of generalized relations (matrices, spans, profunctors; all examples of `cartesian bicategories' \cite{carboni1987cartesian}).  There are various ways to make this connection precise \cite{baillot1997timeless,boudes2009thick,calderon2010understanding,tsukada2016plays,pierre-hugo-lics23} but typically one obtains an oplax functor between the two models. 
The functor is oplax because, despite the similarities, there is a key conceptual difference: the composition of strategies is `interactive' and not `relational', the latter being less constrained. 

This perspective helps to guide our analysis of symmetry in concurrent strategies. In \S\ref{sec:profunctors} we describe an oplax functor, from our model of games and concurrent strategies with symmetry, to a model of small categories and setoid-valued (pseudo) profunctors (Theorem~\ref{thm:oplax-functor}). This functor is especially easy to describe given our characterization of strategies as Street fibrations, using the well-known categorical equivalence between fibred and indexed structures.

The representation of strategies as profunctors is designed to abstract away the complexities of interaction (specific to games and event structures) whilst maintaining the description of symmetries. This abstraction also explains the mechanisms for symmetry in concurrent games in terms of established categorical structures for managing symmetries (setoids, profunctors)

A technical point is that the two models involved are instances of 3-dimensional categories. The 3-cells are necessary, because the 2-cells themselves exhibit symmetry, but we will see that both 3-categories are highly degenerate, so 3-dimensional coherence issues do not arise. We will explain this degeneracy in \S\ref{sec:setoids} by  highlighting the role of setoids (i.e. degenerate groupoids) in the theory of event structures with symmetry. 

\subsection{Summary of contributions and outline of the paper}

Our main contributions are  Theorems~\ref{thm:strong-characterization} and \ref{thm:characterization}, which characterize weak and strong strategies respectively, and Theorem~\ref{thm:oplax-functor}, which establishes the connection to profunctors. 

To state and prove these results we develop some background. In \S\ref{sec:setoids} we give some background on event structures with symmetry, in \S\ref{sec:polarity} we enrich this with polarity and define the `Scott category' of configurations. In \S\ref{sec:fibrations} we give background on categorical fibrations, proving results relevant to our purposes, and define the profunctor model. In \S\ref{sec:games-and-strategies} we finally move to the theory of games and strategies; \S\ref{sec:characterizations} contains the characterization results and \S\ref{sec:strategies-as-profunctors} the connection to profunctors. Finally \S\ref{sec:thin-games} makes a connection to thin concurrent games, a popular approach to symmetry.

\section*{Preliminaries on event structures.} 

An \emph{event
  structure} consists of a set of events $E$ with a partial order
relation $\leq$ of \emph{causal dependency} and a nonempty set $\Con$ of \emph{consistent} finite subsets
of $E$ satisfying four axioms: (1) \emph{finite causes}:
for all $e \in E$, the set $\{ e' \in E \mid e' \leq e \}$ is finite;
(2) for all $e \in E$, the singleton $\{ e \}$ is consistent; (3) if $X \in \Con$ and $Y \subseteq X$ then $Y \in \Con$; (4) if $e
\leq e'$ where $e' \in X \in \Con$, then $X \cup \{ e \} \in \Con$.

A \emph{(finite) configuration} of $(E, \leq, \Con)$ is a finite, consistent,
and down-closed subset of $E$. The set of all such configurations is
  denoted $\config{E}$. We typically use $E$ as shorthand for $(E, \leq,
\Con)$, when no confusion arises. A \emph{(total) map of event structures} $(E, \leq_E, \Con_E) \to (D,
\leq_D, \Con_D)$ is a function $f : E \to D$ satisfying two axioms:
\begin{itemize}
\item \emph{configurations are preserved by direct image}: for all $x \in \config{E}$, $f x \in \config{D}$;
\item \emph{local injectivity}: the restriction of $f$ to any
  configuration $x \in \config{E}$ is injective. 
\end{itemize}

Informally, possible executions of $E$ are all represented via $f$ as
possible executions of $D$. Indeed for every configuration
$x \in \config{E}$, there is a bijection between $x$ and its direct
image: $f_{\mid x} : x \toiso fx$ which is order-reflecting ($f_{\mid x}^{-1}$ is always monotone). In general, a map $f:E\to D$  need not preserve causal dependency; when it does it is called \emph{rigid}. Equivalently, a map $f$ is rigid iff when $x\in\config E$ and $y\in\config D$ and $y\subseteq f(x)$ then there is $z\in\config E$, necessarily unique, for which $z\subseteq x$ and $fz =y$~\cite{ESS}.

In the spectrum of possible models for concurrency, event structures are 
central in that they relate to other models via adjunctions~\cite{handbook,goubault2012formal}. %
They provide a concise, abstract alternative to representations as sets of traces, but at the same time they are more concrete and operational than categorical or topological models for concurrency. Maintaining this compromise is sometimes challenging but holds us close to operational significance and relations with other models. %

\section{Symmetry in event structures and the role of setoids}
\label{sec:setoids}

We review the basic elements of symmetry in event structures
\cite{ESS} and prove a number of new results. One main goal of
this section is to point out the special role played by setoids. (We
will see below in Lemma~\ref{lem:ess-setoid-fibres} that a map of event structures with symmetry induces a setoid structure on each fibre.)

\begin{definition}
\label{def:ess}
  An \emph{event structure with symmetry} consists of an event
  structure $E$ equipped with a family $\sym E = \{
  \sym E [x, y]\}_{x, y \in \config{E}}$, where each $\sym E[x, y]$ is a
  set of bijections $x \xtosym {\ } {} y$, satisfying the following
  axioms:
  \begin{itemize}
  \item \emph{Congruence}: The family $\sym E$ contains all
    identity bijections, and it is closed under composition of
    bijections, inverses, and restriction of a bijection to a
    sub-configuration.
    \item \emph{Bisimulation}: for every $\theta \in \sym E[x, y]$, if
      $x \subseteq x' \in \config{E}$, then there exists $y' \in
      \config{E}$ and $\theta' \in \sym E [x', y']$ such that
      $\theta'_{\mid x} = \theta$.
    \end{itemize}
\end{definition}

\begin{remark}\label{rem:spans} Event structures with symmetry arise from the combination of two abstract ideas: internal equivalence relations as
  spans in a category (e.g. \cite{johnstone2002sketches}), and bisimulation as open maps
  \cite{JNW}. Indeed, to give an event structure with symmetry
  $(E, \sym E)$ is to give a jointly monic span
  $E \xleftarrow{l} R \xrightarrow{r} E$ satisfying the axioms for an
  internal equivalence relation in the category of event structures
  and in which the maps $l$ and $r$ are open. (\emph{Open} means for
  $l$ that it is rigid and if $x \in \config{R}$ and $y \in \config{E}$ with
  $l x \subseteq y$ then there is $z \in \config{R}$ such that
  $x \subseteq z$ and $l z = y$.) The notions of maps and homotopy we
  give below, as well as games with symmetry, could all be developed
  in this style.
\end{remark}

\begin{definition}
A \emph{map of event structures with symmetry} $f : (E, \sym E) \to (D, \sym
D)$ is a map of event structures $f : E \to D$ such that, for every
$\theta \in \sym E [x, y]$, the composite below lies in $\sym D[fx, fy]$:
\[
f \theta \quad := \quad f  x \xiso{f_{\mid x}^{-1}} x \xtosym{\theta} E y \xiso{f_{\mid y}} fy.
\]
\end{definition}

Symmetry gives a notion of 2-cell between maps of event structures. For
maps $f, g : (E, \sym E) \to (D, \sym D)$ of event structures with
symmetry, a \emph{homotopy} $f \sim g$ consists of a symmetry bijection $\theta_x : fx
\tosym D
gx$ for every $x \in \config{E}$, such that
\[
 \begin{tikzcd}[row sep=2mm]
	& fx \\
	x &\\
	& gx
	\arrow["{\theta_x}"{inner sep=0.8ex}, "{\rotatebox{90}{$\sim$}}"'{inner sep=0.2ex}, from=1-2, to=3-2]
	\arrow["{f_{\mid x}}", "{\rotatebox{25}{$\sim$}}"'{inner sep=-0.2ex}, from=2-1, to=1-2]
	\arrow["{g_{\mid x}}"',"{\rotatebox{-25}{$\sim$}}"{inner sep=-0.3ex}, from=2-1, to=3-2]
\end{tikzcd}
  \]
  commutes. It is immediate that if such a homotopy exists then it
  must be unique, with $\theta_x = g_{\mid x} \circ f_{\mid x}^{-1}$,
  and so $\sim$ is just a relation on maps. It is an equivalence
  relation, since there are identity homotopies and homotopies can be
  composed, and moreover it is preserved by composition of maps. Therefore:
\begin{proposition}
 There is a 2-category $\ESS$ of event structures with
  symmetry, maps between them, and homotopies as 2-cells.\footnote{In fact, $\ESS$ has the structure of a homotopy category with path and cylinder objects~\cite{lics14}.}
\end{proposition}

Each hom-category $\ESS(\calE, \calD)$ is a 
groupoid with at most one element in each hom-set. This is equivalently presented as a \emph{setoid}: a set with an equivalence
relation. (In fact one could regard the 2-category $\ESS$ simply as a $\Setoid$-enriched
category, where $\Setoid$ denotes the cartesian monoidal category of
setoids and equivalence-preserving functions.)

There is an internal notion of equivalence in $\ESS$: for
event structures with symmetry $\E$ and $\D$, an \emph{equivalence}
$\E \simeq \D$ is a pair of maps $f : \E \to \D$ and $g : \D \to \E$
such that $f \circ g \sim \id_\D$ and $\id_\E \sim g \circ f$.

\begin{definition}
\label{def:cat-config}
  The \emph{category of configurations} $\confsym{\E}$ of an event
  structure with symmetry $\E = (E, \sym E)$ is the sub-category of
  $\Set$ with object set $\config{E}$ and morphisms generated by
  inclusion maps $x \hookrightarrow y$ (where $x \subseteq y$) and
  symmetry bijections $x \tosym E y$.
\end{definition}
By the properties of symmetry, one can show that a function $x \to y$
is a morphism in $\confsym{\E}$ if and only if it can be factored as
$x \tosym{E} z \hookrightarrow y$ for some configuration
$z \in \config{E}$, and moreover this factorization is unique. We now discuss two important operations on event structures with symmetry.

\subsection{Synchronization via pseudo-pullbacks}
\label{subsec:synchronization}

A \emph{pseudo-pullback} (or \emph{homotopy pullback})
is a strict 2-limit which represents a form of synchronization up to
symmetry. In $\ESS$, where 2-cells are an equivalence relation, the
pseudo-pullback of a cospan
$E \xto{f} D \xleftarrow{f'} E'$ consists of an object $P$ with projections
$p : P \to E$ and $p' : P \to {E'}$ satisfying $fp \sim f'p'$,
having the following universal property: for every other $E \xleftarrow{q}
Q \xto{q'} E'$ such that $fq \sim f'q'$, there is a unique $h : Q \to
P$ such that $q = ph$ and $q' = p'h$.

\newcommand{\bijgraph}{\mathrm{graph}} The 2-category $\ESS$ has all
pseudo-pullbacks \cite{ESS}. For a cospan
$E \xto{f} D \xleftarrow{f'} E'$, we say $x \in \config{E}$ and
$x' \in \config{E'}$ are \emph{synchronizable via
  $\theta : fx \tosym D f'x'$} if there is a partial order on the
graph of the bijection $\theta$ such that the two bijections
$x \toiso \bijgraph(\theta) \lefttoiso x'$ are monotone.\footnote{One can
  always construct a pre-order on $\bijgraph(\theta)$ induced by the
  orders on $x$ and $x'$. The non-automatic property is antisymmetry,
  which corresponds to the absence of deadlocks (or causal loops) in
  the synchronization of $x$ and $x'$.} We can identify the
configurations of the pseudo-pullback $P$ with the sets
$\bijgraph(\theta)$ for tuples $(x, x', \theta)$ where $x$ and $x'$ are
synchronizable via $\theta$. The symmetry $\sym P$ is characterized as
follows: for synchronizable pairs $(x, x', \theta)$ and
$(y, y', \xi)$, a bijection
$\varphi : \bijgraph(\theta) \toiso \bijgraph(\xi)$ is in the symmetry on
$P$ if the dashed bijections below are in the symmetries of $E$ and
$E'$.
\[\begin{tikzcd}
	{x } & {\bijgraph(\theta)} & {x'} \\
	y & {\bijgraph(\xi)} & {y'}
	\arrow["{\sim}"'{outer sep=-1.2ex},from=1-1, to=1-2]
	\arrow["\rotatebox{90}{$\sim$}"{outer sep=-1.2ex}, dashed, from=1-1, to=2-1]
	\arrow["\rotatebox{90}{$\sim$}"'{outer sep=-1.5ex}, "\varphi"', from=1-2, to=2-2]
	\arrow["{\sim}"{outer sep=-1.2ex}, from=1-3, to=1-2]
	\arrow["\rotatebox{90}{$\sim$}"'{outer sep=-1.5ex}, dashed, from=1-3, to=2-3]
	\arrow["{\sim}"'{outer sep=-1.2ex},from=2-1, to=2-2]
	\arrow["{\sim}"{outer sep=-1.2ex},from=2-3, to=2-2]
\end{tikzcd}\]

\subsection{Hiding}
\label{subsec:hiding}

If $(E, \leq, \Con)$ is
an event structure and $V \subseteq E$ is a subset of events, we can
restrict $E$ to $V$. Let $E \downarrow V$ denote the event structure
with events $V$, partial order restricted from $E$, and consistent
subsets the restriction of $\Con$ to subsets of $V$. For
$x \in \config{E \downarrow V}$ we can always complete $x$ to a
configuration $[x]$ of $E$, defined as the minimal
$y \in \config{E}$ such that $y \cap V = x$. Concretely, $[x] = \{ e
\in E\mid \exists e' \in x.\ e \leq e' \}$.

Now let $\sym E$ be a symmetry on $E$. If $V$ is closed under symmetry
(in the sense that for all $\theta \in \sym E[x, y]$, if $e \in x \cap
V$ then $\theta(e) \in V$) then $E \downarrow V$ has a symmetry
consisting of all bijections $\varphi : x \toiso y$ for which there exists
$\psi : [x] \tosym E [y]$ such that the diagram below commutes.
\[\begin{tikzcd}[row sep = 1em]
	x & {[x]} \\
	y & {[y]}
	\arrow[hook, from=1-1, to=1-2]
	\arrow["\varphi"', from=1-1, to=2-1]
	\arrow["\psi", from=1-2, to=2-2]
	\arrow[hook, from=2-1, to=2-2]
      \end{tikzcd}\]

\section{Polarity and the Scott category of configurations}
\label{sec:polarity}

Game semantics begins with the assignment of a polarity to each event,
positive when the event represents a move of Player, and negative for a move
of Opponent.

\begin{definition} An \emph{event structure with
    symmetry and polarity} is an event structure with symmetry
  $\E = (E, \sym E)$ equipped with a labelling function
  $\pol : E \to \{ \boxminus, \boxplus\}$, such that every bijection in $\sym E$
  preserves the polarity of events.
\end{definition}

We denote by $\ESSP$ the 2-category of event structures with symmetry
and polarity, with maps required to preserve polarity and homotopy
2-cells defined as in $\ESS$.

We have seen that an event structure with symmetry has a category of
configurations, a subcategory of $\Set$ containing symmetry
bijections and inclusion maps. We now generalize this to the setting with
polarities in a way that turns out to be useful for characterizing
strategies. First we fix some terminology: for $x, y \in \config{E}$ and
$x \subseteq y$, say the inclusion map $x \hookrightarrow y$ is
\emph{positive} (denoted $x \posincl y$) when all events in
$y \setminus x$ have positive polarity, and \emph{negative}
($x \negincl y$) when they are all negative. Let $\pSet$ denote the
category of sets and partial functions.
\begin{definition}
  If $\E$ is an event structure with polarity and symmetry,
the \emph{Scott category} of $\E$ is the subcategory $\confsym{\E}$ of $\pSet$ with
as objects the elements of $\config{E}$, and morphisms generated by:
\begin{itemize}
\item symmetry bijections $x \tosym E y$;
\item \emph{positive} inclusions maps $x \posincl y$, where $x \subseteq y$ and $y \setminus x$ only contains positive events;
    and
    \item \emph{negative} \emph{reverse} inclusions maps $y \negrev x$, represented by partial
      functions $y \pto x$ acting as identity
      on $x$ and undefined otherwise.
\end{itemize}
\end{definition}

Remarkably, one can show that every morphism $y \pto x$ in the Scott
category $\confsym{E}$ admits a unique factorization of the form
$y \negrev w \toiso u \posincl x$~\cite{lics14}. We will make extensive
use of this fact.

\begin{remark}
  One can see an ordinary event structure with symmetry (without
  polarity) as one with polarity having only positive events. This gives an embedding $\ESS \to \ESSP$. This
  way, the category of configurations of Definition~\ref{def:cat-config}
  becomes a special case of the Scott category.
\end{remark}

The definition of the Scott category already appears in
\cite{lics14}. The novelty in this paper consists in using the Scott
category in a characterization of strategies. This simplifies, and
generalizes, the results of \cite{lics14}. The terminology ``Scott'' is
because of an analogy with the order on functions in domain theory; see
\cite{fossacs13} for further discussion.

We now extend the Scott category construction $\confsym{-}$ to a
2-functor $\ESSP \to \Cat$ acting on maps and homotopies. Given a map
$f : \E \to \D$, recalling that every morphism in $\confsym{\E}$
must be of the form
$x \negrev z \xiso{\theta} w \posincl y$, we can apply $f$ to each component to obtain
$fx \negrev fz \xiso{f\theta} fw \posincl fy$ in $\confsym{\D}$. If $g : \E \to \D$ is another
map with $f \sim g$, then by definition we have a symmetry
$\theta_x : f x \toiso gx$ for every $x \in \confsym{\E}$ and so a
transformation $\confsym{f} \cong \confsym{g}$.

\begin{proposition}
  This definition of $\confsym{-}$ determines a 2-functor, which is
  locally full and faithful. In other words, there is exactly one natural transformation $\confsym{f}\Rightarrow\confsym{g}$ if
  $f \sim g$, and no such transformations otherwise.
\end{proposition}
\begin{proof}
We omit the proof of functoriality.
For local fullness and faithfulness, suppose $\varphi : \confsym{f} \Rightarrow \confsym{g}$ is a
natural transformation. We show that for every $x \in \confsym{E}$
the function $\varphi_x : fx \to gx$ is a symmetry bijection. By
the representation of morphisms in $\confsym{\D}$ we have that
\[
\varphi_x \quad = \quad fx \negrev z \xtosym{\theta} \D w \posincl gx.
\]
By the bisimulation property of $\theta$, there must be $t \in
\config{\D}$ such that $w \negincl t$ and $t
\tosym{\D} fx$. But $f x$ and
$gx$ must have the same number of events of each polarity, so we must
have $w = t = gx$ and $fx = z$, thus $\varphi_x = \theta : fx \tosym \D
gx$.

Now we show that $\varphi_x$ must be the specific bijection
$f(e) \mapsto g(e)$, rather than another symmetry. This is trivially
true if $x = \emptyset$, and by naturality if $\varphi_x$ satisfies
the property and $x \cov^e y$ then $\varphi_y$ also satisfies the
property. The result follows because for every configuration there
must exist at least one $\cov$-chain.
\end{proof}

\textbf{Point of notation.} In the rest of the paper all event structures are equipped with symmetry. So we use traditional letters ($E, A, B$, \dots) to denote them, keeping the symmetries ($\sym E, \sym A, \sym B$, \dots) implicit. Accordingly, $\confsym E$ always refers to the Scott category $\confsym{E, \sym E}$; this should cause no confusion.

\section{Fibrations, Street fibrations, and $\Setoid$-valued profunctors}
\label{sec:fibrations}

The goal of this paper is to study maps of event structures with
symmetry and polarity (called `pre-strategies' from \S\ref{sec:games-and-strategies}) through the fibrational properties of the induced functors on Scott categories.

\begin{definition}
  Let $p : \calC \to \calD$ be any functor and let $D \in \calD$. The
  \emph{fibre} of $p$ over $D$ is the category
  $\fibre p D$ consisting of the objects of $\calC$ that are mapped to
  $D$ and the morphisms of $\calC$ that are mapped to $\id_D$.

The \emph{essential fibre} of $p$ over $D$ is the category $\essfibre
p D$ defined to have as objects pairs
$(C \in \calC, \varphi : pC \toiso D)$ and morphisms
$(C, \varphi) \to (C', \varphi')$ the $f \in \calC(C, C')$ such that
$\varphi' \circ pf = \varphi$.
\end{definition}
The theory of fibrations specifies conditions on
$F$ under which the assignments $D \mapsto \fibre p D$ and $D \mapsto
\essfibre p D$ give pseudofunctors $\calD^\op \to \Cat$.

The next lemma, a key observation for this paper,  shows that for a functor of the form $\confsym{f} : \confsym{\E} \to \confsym{\D}$
between Scott categories, all fibres and essential fibres are setoids.
\begin{lemma}
\label{lem:ess-setoid-fibres}
  Let $f : \E \to \D$ be a map of event structures with symmetry and
  polarity and consider the induced functor $\confsym{f}$. For all
  $x \in \confsym{\D}$, both the fibre and essential fibres over $x$
  are setoids.
\end{lemma}
\begin{proof}
  For any morphism $\alpha : y \pto z$ in $\confsym{\E}$
the diagram
\[\begin{tikzcd}[column sep=1em, row sep=1em]
	y & z \\
	fy & fz
	\arrow["\alpha", harpoon, from=1-1, to=1-2]
	\arrow["{f_{\mid y}}"', from=1-1, to=2-1]
	\arrow["{f_{\mid z}}", from=1-2, to=2-2]
	\arrow["{f\alpha}"', harpoon, from=2-1, to=2-2]
      \end{tikzcd}\] commutes. If $\alpha$ lies in the fibre over $x$,
    then by definition we have $f y = fz = x$ and $f \alpha = \id_x$,
    therefore the commutative diagram implies $\alpha = f_{\mid z}^{-1} \circ f_{\mid y}$. Similar proof for the essential fibre. 
\end{proof}

\subsection{Rudiments of fibrations and setoids}

We recall the general notions in the next definition, see e.g.~\cite[B1.3]{johnstone2002sketches} for a textbook account. The situation will
be greatly simplified when fibres are setoids (Lemma~\ref{lem:setoidfibs}).

\begin{definition}
  For a functor $p : \calC \to \calD$, a morphism
  $f : A \to B$ in $\calC$ is called \emph{cartesian} if for every
  $g : C \to B$ and $u : pC \to pA$ such that $pg = pf \circ u$, there
  is a unique map $v : C \to A$ such that $p v = u$ and $f \circ v =
  g$:
  \[
    \begin{tikzcd}
	C & A & B \\
	pC & pA & pB
\arrow["{v}"{description}, dashed, from=1-1, to=1-2]
	\arrow["g", curve={height=-24pt}, from=1-1, to=1-3]
	\arrow[maps to, from=1-1, to=2-1]
	\arrow["f", from=1-2, to=1-3]
	\arrow[maps to, from=1-2, to=2-2]
	\arrow[maps to, from=1-3, to=2-3]
	\arrow["u"', from=2-1, to=2-2]
	\arrow["pf"', from=2-2, to=2-3]
\end{tikzcd}
\]
We say that $p$ is a \emph{fibration} if for every $B \in \calC$ and $h : X \to pB$
in $\calD$ there is a cartesian morphism $f : A \to B$
with $p f = h$. The morphism $f$ is called a \emph{cartesian lift} of $h$ at $B$.

We say that $p$ is a \emph{Street fibration} if the morphism
$h : X \to pB$ merely has a cartesian \emph{pseudo-lift} at $B$: a
cartesian morphism $f : A \to B$ together with an isomorphism
$\theta : pA \toiso X$ such that $p f = h \circ \theta$, as pictured
below.
\[\begin{tikzcd}
	A && B \\
	pA & X & pB
	\arrow["f", from=1-1, to=1-3]
	\arrow[maps to, from=1-1, to=2-1]
	\arrow[maps to, from=1-3, to=2-3]
	\arrow["\theta"', "\sim", from=2-1, to=2-2]
	\arrow["h"', from=2-2, to=2-3]
\end{tikzcd}\]
\end{definition}
(Fibrations are sometimes called \emph{Grothendieck fibrations}, to
distinguish them from Street fibrations.)

In general, there may be many possible lifts of a morphism $h : X \to
pB$ at an
object $B$, but a cartesian
lift satisfies a universal property which relates it to every other
lift. Cartesian lifts are \emph{essentially unique}: if $f : A \to B$
and $f' : A' \to B$ are cartesian lifts of $h : X \to pB$ then there a
unique isomorphism $\varphi : A \to A'$ such that $p \varphi = \id_X$
and $f = f' \circ \varphi$.  Likewise for pseudo
lifts: if $(f : A \to B, \theta : pA \toiso X)$ and
$(f' : A' \to B, \theta' : pA' \toiso X)$ are cartesian pseudo-lifts
of $h$, then there is a unique isomorphism $\psi : A \to A'$ such that
$p \psi = {\theta'}^{-1} \circ \theta$ and $f = f' \circ \psi$.

We now consider the special case in which a functor $p$ has
essentially unique lifts:
\begin{definition}
  A functor $p : \calC \to \calD$ has \emph{essentially unique lifts}
  when the following condition holds: for every $B \in \calC$ and
  $h : X \to pB$ in $\calD$, if $f : A \to B$ and $f' : A' \to B$
  satisfy $pf = pf' = h$ then there is a unique isomorphism
  $\varphi : A \to A'$ such that $p\varphi = \id_X$ and
  $f' = f \circ \varphi$.

  Similarly, we say that $p : \calC \to \calD$ has \emph{essentially
    unique pseudo-lifts} when the following condition holds: for every
  $B \in \calC$ and $h : X \to pB$ in $\calD$, if
  $(f : A \to B, \theta : pA \toiso X)$ and
  $(f' : A' \to B, \theta' : pA' \to X)$ satisfy $pf = \theta \circ h$
  and $pf' = \theta' \circ h$ then there is a unique isomorphism
  $\varphi : A \to A'$ such that $\theta' = \theta \circ p\varphi$ and
  $f' = f \circ \varphi$.
\end{definition}
When lifts are essentially unique they are automatically cartesian, the
functor is automatically a fibration, and all fibres are setoids; and
similarly for pseudo-lifts.
\begin{lemma}
  \label{lem:setoidfibs}
  The following are equivalent for a functor $p : \calC \to \calD$.
  \begin{itemize}
  \item $p$ has essentially unique lifts (resp. pseudo-lifts).
  \item $p$ is a fibration (resp. Street fibration) and every fibre
    (resp. essential fibre) is a setoid.
  \end{itemize}
  When these conditions are satisfied, we call the functor $p$ a
  fibration in setoids (resp. a Street fibration in setoids).
\end{lemma}
\begin{proof}
Elementary verification.
\end{proof}

We now discuss the relationship of fibrations and Street fibrations in
setoids with \emph{indexed} setoids.
\begin{definition}
  A fibration-in-setoids $p : \calC \to \calD$ is \emph{cloven} when
  it is equipped with a choice of lift $f^*C \to C$ for every $C \in
  \calC$ and $f : D \to pC$. Similarly, a Street fibration-in-setoids
  is \emph{cloven} when equipped with a choice of pseudo-lifts.
\end{definition}

Assuming the axiom of choice, such a choice of lifts can always
be made for a fibration. We consider the following 2-categories for a
small category $\calD$:
\begin{itemize}
\item $\SetoidSFib(\calD)$ is the 2-category with objects pairs
  $(\calC, p : \calC \to \calD)$ where $\calC$ is a small category and
  $p$ is a Street fibration-in-setoids. The morphisms
  $(\calC, p) \to (\calC', p')$ are pairs consisting of a functor
  $F : \calC \to \calC'$ and a natural isomorphism
  $\varphi : p \Rightarrow p' \circ F$, and the 2-cells $(F, \varphi)
  \to (F, \varphi)$.

\item $\SetoidFib(\calD)$ is the sub-2-category of $\SetoidSFib(\calD)$
  consisting of cloven fibrations-in-setoids, strict
  morphisms (with $p = p' \circ F$ and $\varphi = \id$), and all
  2-cells between them.
\item $[\calD^\op, \Setoid]$ is the 2-category of pseudo-functors
  $\calD^\op \to \Setoid$. Morphisms are pseudo-natural
  transformations (with the naturality square commuting up to
  equivalence), and 2-cells are component-wise equivalence of pseudo
  natural transformations.
\end{itemize}

The fibres of a cloven fibration, and the essential
fibres of a cloven Street fibration, can be presented functorially.
Indeed there are 2-functors
  \begin{align*}
&    \fibres : \SetoidFib(\calD) \longto [\calD^\op, \Setoid] \\
&\essfibres : \SetoidSFib(\calD) \longto [\calD^\op, \Setoid]
  \end{align*}
(Via the Grothendieck construction  \cite[B1.3]{johnstone2002sketches},  both are actually
equivalences of 2-categories; $\fibres$ is even a strict
equivalence.)

\subsection{A 3-dimensional profunctor model}
\label{sec:profunctors}

We introduce a model of profunctors valued in setoids. Although setoids are categorically
equivalent to sets, this level of structure is important for us because it contains just the right amount of information to model the symmetries of event structures.

The profunctor model is formally a tricategory, but the structure at level 3 is very simple because the
3-cells encode an equivalence relation on the 2-cells. In particular
there are no three-dimensional coherence axioms to verify. (The same will apply to games and strategies in \S\ref{sec:games-and-strategies}.)

\begin{definition}
The tricategory $\SetoidProf$ is given by the following components.
\begin{itemize}
\item Objects are small categories ($\calC, \calD, \dots$).
\item For objects $\calC$ and $\calD$, $\SetoidProf[\calC, \calD]$ is the 2-category
    $[\calD^\op \times \calC, \Setoid]$ of pseudo-functors, pseudo
    natural transformations, and equivalences. We write $P : \calC
    \profto \calD$ if $P \in \SetoidProf[\calC, \calD]$.
\item The identity on $\calC$ is the hom-functor
   $\calC(-, =) : \calC^\op \times \calC \to \Set$ post-composed with
    the embedding $\Set \to \Setoid$ that regards a set as a discrete
    setoid.
\item The composition of $P : \calC \profto \calD$ and $Q : \calD \profto
  \calE$ is the profunctor $Q \odot P : \calE^\op \times \calD \to
  \Setoid$ whose action on objects is given by
  \[
(Q\odot P)(E, C) = \int^{D \in \calD} P(D, C) \times Q(E, D)
\]
where the integral sign denotes a pseudo-coend in $\Setoid$, constructed as
the set $\sum_{D \in \calD} P(D, C) \times Q(E, D)$ with the smallest
equivalence relation $\sim$ such that
\begin{align*}
& (D, p, q) \sim (D, p', q)  && \text{when } p \sim p' \text{ in }
  P(D, C), \\
& (D, p, q) \sim (D, p, q')  & &\text{when } q \sim q' \text{ in }
                               Q(E, D), \\
 & (D, P(\alpha, C)(p'), q) \sim (D', p', P(E, \alpha)(q))  &&
                                                             \text{for
                                                             all } \alpha
                                                             \in
                                                             \calD(D',
                                                             D), \\
&&&\quad  p' \in P(D', C), q \in Q(E, D)
                                                             .
\end{align*}
We omit the action of the functor $Q\circ P$ on morphisms, as well as
the horizontal composition of 2-cells and 3-cells.
\end{itemize}
\end{definition}

Pseudo-profunctors of this kind are not new to this paper. See the references \cite{Lawler2014,Chikhladze2015LaxFormalTheory} for fully-detailed presentations, albeit in greater generality.

\begin{remark}
  We emphasize that $\SetoidProf$ is not a fully weak tricategory. The hom-2-categories are all
  strict and so the ``vertical'' composition of 2-cells and 3-cells is
  strictly associative and unital. However, the tricategory is not
  fully strict because the composition of profunctors is only
  associative and unital up to equivalence, where an \emph{equivalence}
  is a pair of 2-cells which are inverses of each other up to a 3-cell.
\end{remark}

\section{Games and strategies with symmetry}
\label{sec:games-and-strategies}

In this section we consider games and strategies as event structures with symmetry. We define identity strategies and the composition operation. These notions already appear in \cite{lics14}, but we generalize the model of \cite{lics14} to also include `weak' strategies, that satisfy the laws of composition only up to weak equivalence. 

\subsection{Games and pre-strategies} 
\label{subsec:pre-strategies}

The basic setup is simple to
describe: a \emph{game} is an event structure with symmetry and
polarity, and a \emph{pre-strategy} on a game $A$ is an event
structure with symmetry and polarity $S$ equipped with a map
$\sigma : S \to A$. We will later define a strategy as a pre-strategy
for which composition works well; see Definition~\ref{def:strategy}.

All reasoning about strategies and pre-strategies will be up to
equivalence, where equivalence is defined as follows. For pre-strategies
$\sigma : S \to A$ and $\sigma' : S' \to A$, a \emph{map of
  pre-strategies $\sigma \to \sigma'$} is a map $f : S \to S'$ such
that $\sigma' \circ f \sim \sigma$. The map $f$ is called
\emph{strict} if $\sigma' \circ f = \sigma$. A \emph{weak equivalence} (or just \emph{equivalence})
$\sigma \weakequiv \sigma'$ is a pair of maps $f : \sigma \to \sigma'$ and
$g : \sigma' \to \sigma$ which form an equivalence $S \simeq S'$ in $\ESS$. It
is a \emph{strong equivalence} if $f$ and $g$ are strict maps (in which case we write $\sigma \strongequiv \sigma'$), an
\emph{isomorphism} if $f$ and $g$ are inverses of each other, and a
\emph{strong isomorphism} if they are both strict and inverses.

\begin{example} To illustrate, consider a simple game $A$
  consisting of two positive, symmetric, and consistent events. We
  might draw $A$ as:
  
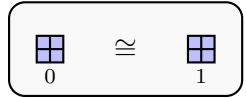
\[ 
  \begin{tikzpicture} 
  \node [posnode=0] (a) at (0, 0) {}; 
  \node [posnode=1] (b) at (2,0) {};
  \node (sym) at (1, 0) {$\cong$};
\begin{scope}[on background layer]
\node[
  draw,
  rounded corners=6pt,
  thick,
  fill=gray!5,
  inner sep=10pt,
  fit=(a) (b) (sym)
] {};
\end{scope}
   \end{tikzpicture}
   \]
     where `$\cong$' indicates a symmetry between singleton configurations (and an auto-symmetry for the two-event configuration). The following are six possible pre-strategies $S$ on this game. In each case the map $S \to A$ is the only function preserving the labels. Recall that a squiggly line between events indicates conflict (\ie\ any set containing both is inconsistent):
  \tikzset{every picture/.style={baseline}}
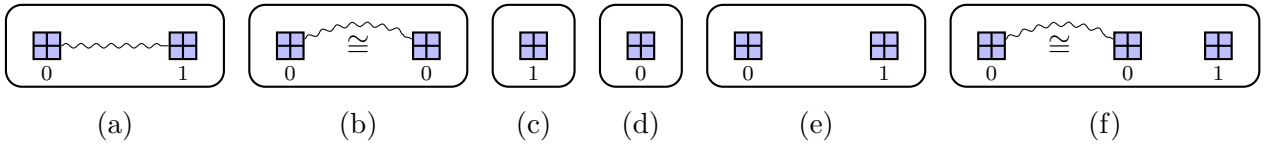
\[
\begin{tikzpicture}
\node[posnode=0] (a) at (0,0) {};
\node[posnode=1] (b) at (1.8,0) {};
\draw[conflict] (a) -- (b);
\node[frame,fit=(a)(b)] {};
\node[below=5pt] at (current bounding box.south) {(a)};
\end{tikzpicture}
\hspace{0.8em}
\begin{tikzpicture}
\node[posnode=0] (a) at (0,0) {};
\node[posnode=0] (b) at (1.8,0) {};
\node (sym) at (0.9, 0) {$\cong$};
\draw[conflict,bend left=25] (a) to (b);
\node[frame,fit=(a)(b)] {};
\node[below=5pt] at (current bounding box.south) {(b)};
\end{tikzpicture}
\hspace{0.8em}
\begin{tikzpicture}
\node[posnode=1] (a) at (0,0) {};
\node[frame,fit=(a)] {};
\node[below=5pt] at (current bounding box.south) {(c)};
\end{tikzpicture}
\hspace{0.8em}
\begin{tikzpicture}
\node[posnode=0] (a) at (0,0) {};
\node[frame,fit=(a)] {};
\node[below=5pt] at (current bounding box.south) {(d)};
\end{tikzpicture}
\hspace{0.8em}
\begin{tikzpicture}
\node[posnode=0] (a) at (0,0) {};
\node[posnode=1] (b) at (1.8,0) {};
\node[frame,fit=(a)(b)] {};
\node[below=5pt] at (current bounding box.south) {(e)};
\end{tikzpicture}
\hspace{0.8em}
\begin{tikzpicture}
\node[posnode=0] (a) at (0,0) {};
\node[posnode=0] (b) at (1.8,0) {};
\node[posnode=1] (c) at (3,0) {};
\node (sym) at (0.9, 0) {$\cong$};
\draw[conflict,bend left=25] (a) to (b);
\node[frame,fit=(a)(b)(c)] {};
\node[below=5pt] at (current bounding box.south) {(f)};
\end{tikzpicture}
\]
The pre-strategies (b) and (d) are strongly equivalent to each other and weakly equivalent to (c), while (e) and (f) are strongly equivalent. (No other equivalences hold.)
\end{example}

\newlength{\mylen}
\newlength{\mywid}
\settoheight{\mylen}{$t$}
\settowidth{\mywid}{$A$}
\renewcommand{\parallel}{\mathbin{\raisebox{0.15em}{\resizebox{\mywid}{\mylen}{$\Vert$}}}}

\subsection{The composition of pre-strategies}
For games $A$ and $B$, a
pre-strategy \emph{from $A$ to $B$} is a pre-strategy on the game
$A^\perp \parallel B$, where $(-)^\perp$ inverts the polarity of every
event, and $\parallel$ is the parallel composition operator on event
structures with symmetry. (The latter is defined simply as a disjoint union of
events with all structure inherited component-wise.)
For strategies of the form $\sigma : S \to A^\perp \parallel B$, we
sometimes keep $S$ implicit and write $\sigma : A \profto B$.

Pre-strategies $\sigma : S \to A^\perp \parallel B$ and
$\tau : T \to B^\perp \parallel C$ are composed in two steps. The
first step is to describe the interaction of $\sigma$ and $\tau$ via  the homotopy pullback
\[\begin{tikzcd}[row sep=0em]
	& {S \inter T} \\
	{S \parallel C} && {A \parallel T} \\
	& {A \parallel B \parallel C}
	\arrow["{\pi_1}"', from=1-2, to=2-1]
	\arrow["{\pi_2}", from=1-2, to=2-3]
	\arrow["\scalebox{2}{$\lrcorner$}"{anchor=center, pos=0.125, rotate=-45}, draw=none, from=1-2, to=3-2]
	\arrow["{\sigma \parallel C}"', from=2-1, to=3-2]
	\arrow["{A \parallel \tau}", from=2-3, to=3-2]
      \end{tikzcd}\]
in $\ESS$, where the maps  $S \to A \parallel B$
and $T \to B \parallel C$ are the morphisms underlying the two
pre-strategies.
The event structure with symmetry $S \inter T$ is called the \emph{interaction}
of $\sigma$ and $\tau$ and admits two maps $S \inter T \to A \parallel B \parallel C$ which are equivalent. We pick one for preciseness: let $\sigma \inter \tau$ denote the left way around the diagram.

The second step is to ``hide'' the events of $B$, by restricting $S
\circledast T$ to the subset of events mapped to $A \parallel C$.
This gives an event structure with symmetry denoted $S \comp T$. There is an unambiguous way to add polarities back to obtain a
pre-strategy $\sigma \comp \tau : S \comp T \to A^\perp \parallel C$,
called the \emph{composition} of $\sigma$ and $\tau$.

This composition is associative up to equivalence, and functorial with respect to maps of strategies. In particular, weak equivalences $\sigma \weakequiv \sigma'$ and $\tau \weakequiv \tau'$ induce a weak equivalence $\tau \comp \sigma \weakequiv \tau' \comp \sigma'$; and the same holds if $\weakequiv$ is replaced with $\strongequiv$. (For proofs of these facts see \cite{castellan2017concurrent,lics14}.)

\subsection{Copycat}
\label{subsec:copycat}

Every game $A$ admits a ``copycat'' strategy $A \profto A$, to serve as the identity morphism. This will be denoted $\cc_A : \CC_A  \to A^\perp \parallel A$. Both $\CC_A$ and $\cc_A$ are characterized by the next lemma, which we have formulated in such a way that the connection with the profunctor model of Section~\ref{sec:profunctors} is manifest. (The explicit construction of copycat can  be found in \cite{lics14}; it is not required for this paper.)

\newcommand{\Tw}{\mathrm{Tw}}
Recall that, for a category $\calC$, the \emph{twisted arrow category} $\Tw(\calC)$ is obtained as the category of elements of the hom-functor $\calC^\op \times \calC \to \Set$. That is, objects are arrows $(f : c \to c')$ in $\calC$ and morphisms $(f_1 : c_1 \to c_1') \to (f_2 : c_2 \to c_2')$ are pairs of morphisms $g : c_2 \to c_1$ and $g' : c_1' \to c_2'$ such that $f_2 = g' \circ f_1 \circ g$.

\begin{lemma}[Configurations of copycat, \cite{lics14}]
\label{lem:conf-copycat}
For a game $A$, there exists an event structure with symmetry and polarity $\CC_A$ such that the category $\confsym{\CC_A}$ is isomorphic to the  category $\Tw(\confsym{A})$. Furthermore there is a map of event structures with symmetry $\cc_A : \CC_A \to A^\perp \parallel A$, corresponding (via $\confsym{-}$ and the mediating isomorphisms) to the canonical projection $\Tw(\confsym{A}) \to \confsym{A}^\op \times \confsym{A}$.
\end{lemma}

\begin{lemma}[Composition with copycat, adapted from \cite{lics14}]
For every pre-strategy $\sigma : A \profto B$ there are strict maps of strategies $\lambda_\sigma : \sigma \to \cc_B \comp \sigma$ and $\rho_\sigma : \sigma \to \sigma \comp \cc_A$, natural in $\sigma$.
\end{lemma}
In the next section, we will study in detail what happens when a pre-strategy is composed with copycat, and in particular we will recall the concrete constructions of $\lambda_\sigma$ and $\rho_\sigma$.

\subsection{Strategies}

\begin{definition}
\label{def:strategy}
  A \emph{(weak) strategy from $A$ to $B$}
  is a pre-strategy
  $\sigma : S \to A^\perp \parallel B$ for which the canonical maps
  $\lambda_\sigma$ and $\rho_{\sigma}$ are weak equivalences of pre-strategies. Say $\sigma$ is a \emph{strong strategy} if they are strong equivalences.\footnote{In this paper by `strategy' we always mean weak strategy, unless strong is specified. This emphasizes that they are the more general notion. The reference \cite{lics14} instead uses strong strategies as default, specifying `weak' otherwise.}
\end{definition}
It follows directly from this definition (and the associativity of composition) that strategies are closed under
weak equivalence and under composition. We organize all of this data into a tricategory of games and strategies. The level of strictness is
the same as for $\SetoidProf$: the composition of 2-cells is strictly
associative and unital, and the 3-cells are just a congruence relation on the 2-cells. We emphasize that nonetheless a three-dimensional structure is needed, because copycat is only an identity up to equivalence of strategies, and the definition of (both weak and strong) equivalence relies on the 3-cells.
\begin{definition}
  The tricategory $\Strat$ is given by the following components.
\begin{itemize}
 \item Objects are games $(A, B, \dots)$.
 \item For objects $A$ and $B$, $\Strat[A, B]$ is the 2-category whose
   objects are strategies from $A$ to $B$; morphisms are maps of
   strategies; and 2-cells are homotopies between maps.
 \item The identity on $A$ is the copycat strategy $\cc_A$.
 \item The composition 2-functor
   $\Strat[B, C] \times \Strat[A, B] \to \Strat[A, C]$ extends the
   composition $\comp$ to maps of strategies and their equivalences,
   via the universal property of pseudo-pullbacks.
\end{itemize}
\end{definition}
\newcommand{\StrongStrat}{\Strat_{\mathrm{strong}}}
There is a sub-tricategory consisting of games, strong strategies, strict maps between them, and homotopies. This sub-model, which we denote by $\StrongStrat$, is the `$\sim$-bicategory' developed in \cite{lics14}. 

\begin{remark}
  Tricategories are generally hard to construct because of the many coherence axioms, but in $\Strat$ all axioms involving equations of 3-cells hold automatically because the structure is thin. So it is enough to check that the axioms of a bicategory hold up to the equivalence relation represented by the 3-cells. This is much simpler and follows from our description of composition using pullback and hiding. (We emphasize that these axioms \emph{do not} hold up to equality, and that a three-dimensional structure is really needed to describe the compositional framework.)
\end{remark}

In this section we have presented a compositional model of games
and strategies up to equivalence. This is in some sense the most
general compositional model based on $\ESS$: strategies are, by definition, exactly the pre-strategies for which composition works
well. Our goal in this paper is to provide an alternative, more direct characterization of strategies.

\section{Characterizations of weak strategies and strong strategies}
\label{sec:characterizations}

The central result of this section is that weak strategies correspond to Street fibrations. To establish this, we revisit the possible lifting properties that pre-strategies may satisfy (\S\ref{subsec:lifting}), several of which have previously been used to carve out classes of well-behaved pre-strategies. We then introduce an important `collapse' construction on strategies (\S\ref{subsec:collapse}) which allows us to characterize first the strong strategies (\S\ref{subsec:strong-characterization}) and then the weak ones (\S\ref{subsec:weak-characterization}). 

\begin{remark}
Throughout this section we consider pre-strategies on a game $A$, rather than $A^\perp \parallel B$. We regard them as morphisms $\varnothing \profto A$ in $\Strat$ and study the post-composition with copycat on $A$. All our definitions and results apply equally to pre-strategies $\sigma : A \profto B$ with completely analogous arguments. Note in particular that $\cc_B \comp \sigma \comp \cc_A \strongequiv \cc_{A^\perp \parallel B} \circ \sigma$. We stick to the special case for readability.
\end{remark}

\subsection{Lifting properties for pre-strategies}
\label{subsec:lifting}

The following properties will play a role in the technical development. Pseudo-receptivity is new to this paper. 
\begin{definition}
\label{def:lifting-properties}
A pre-strategy $\sigma : S\to A$ is said to be:
\begin{itemize}
\item \emph{innocent} (a.k.a.~\emph{courteous} in \cite{Asgames,LMCS,cg2}) if, for every $x \in \confsym{S}$,  every positive inclusion $y \posincl \sigma x$  in $\config{A}$ has a (necessarily unique, by local injectivity) %
lift;
\item \emph{receptive} if, for every $x \in \confsym{S}$, there exists a unique %
lift of every negative inclusion $y \negrev \sigma x$  in $\config{A}$.
\item \emph{pseudo-receptive}
if, for every $x \in \confsym{S}$, every negative inclusion $y \negrev \sigma x$  in $\config{A}$ has an essentially unique pseudo-lift; 
\end{itemize}
The condition known as \emph{strong receptivity} in \cite{lics14} is the conjunction of receptivity and pseudo-receptivity, which are independent conditions. 
\end{definition}

\begin{remark}  One could equivalently state innocence in terms of \emph{immediate causality}: for every $s, s' \in S$, if $s$ has positive polarity and $s \imc s'$ then $\sigma s \imc \sigma s'$, where $\imc$ means $<$ with no events in between. For a proof of equivalence see \cite{fossacs13}. (We also note an unfortunate terminology clash: innocence %
is not related to the notion of innocent strategy in Hyland-Ong game semantics \cite{HO}.)
\end{remark}

The following lemma will be important in our proofs:
\begin{lemma}
  \label{lem:minusinnocence}
For a pseudo-receptive pre-strategy $\sigma : S \to A$, the following \emph{$\boxminus$-innocence} condition holds: if $s, s \in S$ with $s \imc s'$ and both $s, s'$ are negative, then $\sigma s \imc \sigma s'$. 
\end{lemma}
Recall that in an event structure the notation $x \cov^e y$ means that $e \notin y$ and $y = x \cup \{e\}$.
\begin{proof}
 We sketch a proof of the lemma. Necessarily $x \cov^s x_1 \cov^{s'} x'$ for some $x, x_1, x' \in \confsym{S}$. Suppose for a contradiction that $s \imc s'$ and $\sigma s$ and $\sigma s'$ are concurrent, so both $\sigma x \cov^{\sigma s} \sigma x_1 \cov^{\sigma s'} \sigma x'$ and also $\sigma x \cov^{\sigma s'} y \cov^{\sigma s} \sigma x'$ for some $y \in \confsym{A}$. By applying pseudo-receptivity to the inclusions $\sigma x \hookrightarrow y \hookrightarrow \sigma x'$ we obtain $x \cov^{t'} x_2 \cov^{t} x''$ and a symmetry $\theta : \sigma x'' \cong \theta x'$ where $\theta$ fixes $\sigma x$ and acts by $\sigma t' \mapsto \sigma s'$ and $\sigma t \mapsto \sigma s$. We now have two pseudo-lifts of the inclusion $\sigma x \hookrightarrow \sigma x'$, namely $x \hookrightarrow x'$ and $x \hookrightarrow x''$, so there must be a symmetry $\varphi : x \cong x'$ which maps to $\theta$. Necessarily $\varphi$ acts by $s \mapsto t$ and $s' \mapsto t'$. Since $s \imc s'$ we must have $t \imc t'$ which is a contradiction as $t'$ comes before $t$ in the chain $x \cov^{t'} x_2 \cov^{t} x''$. 
\end{proof}

Next we consider a lifting condition for symmetries, also imported from \cite{lics14}.
\begin{definition}
If $\sigma : S \to A$ is a pre-strategy, its \emph{saturation}  is the pre-strategy $\sat{\sigma} : \sat{S} \to A$ obtained as the right projection in the pseudo-pullback below.
  \[\begin{tikzcd}[row sep=0em]
	& {\sat{S}} & \\
	S && A \\
	& A
	\arrow["\pi_1"',pos=0.125, from=1-2, to=2-1]
	\arrow["{\sat{\sigma}}", pos=0.125,from=1-2, to=2-3]
	\arrow["\sigma"', from=2-1, to=3-2]
  \arrow["\scalebox{2}{$\lrcorner$}"{anchor=center, pos=0.125, rotate=-45}, draw=none, from=1-2, to=3-2]
	\arrow["{\id_{A}}", from=2-3, to=3-2]
\end{tikzcd}\]
\end{definition}
The saturation process can be understood as freely adding to $\confsym{S}$ a lift for every
isomorphism in $\confsym{A}$, in a way that is compatible with the causal structure of $S$ (in the sense that this free completion is performed internally in $\ESS$). Indeed, configurations of $\sat{S}$ corresponds to pairs of a configuration $x$ of $S$ and a symmetry $\sigma x \cong y$ and such a pair is mapped to $y$ under $\sat{\sigma}$. There is a canonical, strict map of strategies $\eta_\sigma : \sigma \to \sat{\sigma}$ of pre-strategies
 such that $\eta_\sigma x$ is the configuration corresponding to the pair $(x, \id_{\sigma x})$. 
 (The existence of $\eta_\sigma$ also follows directly from the universal property of the pseudo-pullback.)
 
A pre-strategy is said to be \emph{saturated} if
$\eta_\sigma$ is a strong equivalence of pre-strategies, i.e. it admits a pseudo-inverse $\mathrm{act}: \sat{\sigma} \to \sigma$ which is a strict map. The existence of $\mathrm{act}$ makes $\confsym{\sigma}$ an \emph{isofibration} but the converse does not hold: some isofibrations are not saturated because the isomorphism lifts do not organize themselves into a map of event structures. Here is an example.

\begin{example}
\label{ex:isofib-not-saturated}
The following pre-strategy determines a Grothendieck fibration (so in
particular an isofibration) between the induced Scott categories, but an
easy argument shows that it is not saturated.
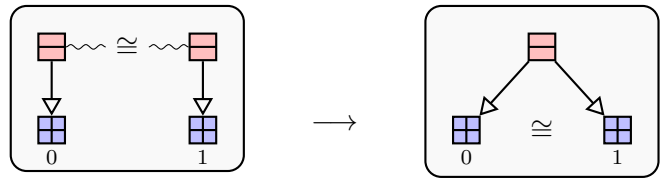
\[
\begin{tikzpicture}[baseline]
\node[negnode] (n0) at (0,1.1) {};
\node[negnode] (n1) at (2,1.1) {};
\node (symn) at (1,1.1) {$\cong$};
\draw[conflict] (n0) -- (symn);
\draw[conflict] (symn) -- (n1);
\node[posnode=0] (p0) at (0,0) {};
\node[posnode=1] (p1) at (2,0) {};
\draw[strat-causality] (n0) -- (p0);
\draw[strat-causality] (n1) -- (p1);
\begin{scope}[on background layer]
\node[frame,fill=gray!5,fit=(n0)(n1)(p0)(p1)] {};
\end{scope}
\end{tikzpicture}
\qquad\longrightarrow\qquad
\begin{tikzpicture}[baseline]
\node[negnode] (m) at (1,1.1) {};
\node[posnode=0] (q0) at (0,0) {};
\node[posnode=1] (q1) at (2,0) {};
\node (symq) at (1,0) {$\cong$};
\draw[strat-causality] (m) -- (q0);
\draw[strat-causality] (m) -- (q1);
\begin{scope}[on background layer]
\node[frame,fill=gray!5,fit=(m)(q0)(q1)(symq)] {};
\end{scope}
\end{tikzpicture}
\]
\end{example}

In general, note that the saturation of a pre-strategy is itself saturated: this follows from elementary reasoning with pseudo-pullbacks. We shall make use of the fact that any pre-strategy is weakly equivalent to its saturation. This observation appears to be new to this paper. 
\begin{lemma}\label{lem:satweakequiv}
  For any pre-strategy $\sigma : S \to A$,  $\sigma$ and $\sat{\sigma}$ are weakly equivalent via the maps $\eta_\sigma$ and $\pi_1$.
\end{lemma}
\begin{proof}
  Elementary reasoning using the universal property of the pseudo-pullback.
\end{proof}

Next we state the key result of \cite{lics14}, that gives sufficient conditions on a pre-strategy for the unit laws to hold \emph{up to strong equivalence}.
\begin{theorem}[\cite{lics14}]
\label{thm:old-characterization}
  Let $\sigma : S \to A$ be a pre-strategy which is saturated, innocent, and strongly receptive (i.e. both receptive and pseudo-receptive). Then $\sigma$ is a strong strategy, i.e. the map $\lambda_\sigma : \sigma \to \cc_A \odot \sigma$ is a strong equivalence of pre-strategies. 
\end{theorem}
These conditions are sufficient but not necessary. Below we will give a complete characterization of strong strategies, via a slightly weaker set of conditions (Theorem~\ref{thm:strong-characterization} in \S\ref{subsec:strong-characterization}).

\begin{lemma}
\label{lem:stability}
The lifting conditions of Definition~\ref{def:lifting-properties} satisfy the stability properties we list below.
\begin{enumerate} 
\item Innocent pre-strategies are closed under weak equivalence. 
\item Pseudo-receptive pre-strategies are closed under weak equivalence.
\item Saturated pre-strategies are closed under strong equivalence, but not weak equivalence.
\end{enumerate}
\end{lemma}

\subsection{The collapse of a pre-strategy}
\label{subsec:collapse}

Our first step is to describe an operation on pre-strategies which collapses `redundant' lifts of negative inclusions. This is an essential construction for this paper.

\begin{theorem}
\label{thm:collapse}
Let $\sigma : S \to A$ be an innocent and pseudo-receptive pre-strategy. There is a pre-strategy $\collapse(\sigma) : \collapse(S) \to A$, strongly equivalent to $\sigma$, which satisfies the uniqueness part of the receptivity condition.
\end{theorem}

We construct $\collapse(\sigma)$ using the machinery of stable families (Appendix~\ref{app:stable-families}), convenient for constructing event structures and symmetries. A map $\sigma \to \collapse(\sigma)$ is obtained easily and canonically, but our construction of the pseudo-inverse $\collapse(\sigma) \to \sigma$ makes an essential use of Zorn's lemma, equivalent to the axiom of choice. All details of this construction are given in Appendix~\ref{app:collapse}.

\begin{corollary}
\label{cor:collapse}
Suppose $\sigma : S \to A$ is an innocent, pseudo-receptive strategy which satisfies the existence part of the receptivity condition. Then $\collapse(\sigma)$ is innocent, pseudo-receptive, and receptive. 
\end{corollary}
\begin{proof}
Innocence and pseudo-receptivity are preserved by strong equivalence (Lemma~\ref{lem:stability}), and both the existence and uniqueness parts of receptivity hold for $\collapse(\sigma)$.
\end{proof}

\subsection{A characterization of strong strategies}
\label{subsec:strong-characterization}

\begin{lemma}
\label{lem:existence-part-for-free}
   If a pre-strategy $\sigma : S \to A$ is pseudo-receptive then $\Sat(\sigma)$ satisfies the existence part of receptivity.
\end{lemma}
\begin{proof}
  The map $\Sat(\sig)$ sends a configuration $w \in \config{\sat{S}}$,
  corresponding to a pair
  $(x \in \config{S}, \sig x \xtosym{\theta}{} y)$, to $y$. If
  $y \negincl y'$ in $\config{A}$, then by the properties of symmetry
  there is $y''$ such that
\[\begin{tikzcd}[row sep=1.8em, column sep=1.8em]
	{y''} & {y'} \\
	{\sig x} & y
	\arrow["{\sim}"'{outer sep=-1.2ex}, from=1-1, to=1-2]
	\arrow["{\scriptscriptstyle\boxminus}"', hook, from=2-1, to=1-1]
	\arrow["{\sim}"'{outer sep=-1.2ex}, from=2-1, to=2-2]
	\arrow["{\scriptscriptstyle\boxminus}", hook, from=2-2, to=1-2]
\end{tikzcd}\]
 commutes. Assuming $\sig$ is pseudo-receptive, there is
 $x' \in \config{S}$ with $x \negincl x'$ with a symmetry
 $\sig x' \toiso y''$ making the following commute:
\[\begin{tikzcd}[row sep=1.8em, column sep=1.8em]
	{\sig x'} & {y''} & {y'} \\
	& {\sig x} & {y\,.}
	\arrow["{\sim}"'{outer sep=-1.2ex}, from=1-1, to=1-2]
	\arrow["{\sim}"'{outer sep=-1.2ex}, from=1-2, to=1-3]
	\arrow["{\scriptscriptstyle\boxminus}"', hook, from=2-2, to=1-1]
	\arrow["{\scriptscriptstyle\boxminus}"', hook, from=2-2, to=1-2]
	\arrow["{\sim}"'{outer sep=-1.2ex}, from=2-2, to=2-3]
	\arrow["{\scriptscriptstyle\boxminus}", hook, from=2-3, to=1-3]
\end{tikzcd}\]
  The pair $(x', \sig x' \toiso y')$ corresponds to a configuration $w'$
  of $\sat{S}$ which is sent to $y'$ under $\Sat(\sig)$.
  Now, by a straightforward characterization of the inclusion order in
  $\config{\sat{S}}$, it holds that $w \subseteq w'$ and
  $\Sat(\sig)(w') = y'$, showing the existence part of receptivity for
  $\Sat(\sig)$.
  \end{proof}

\begin{theorem}
\label{thm:strong-characterization}
A pre-strategy $\sigma : S \to A$ is a strong strategy if and only if it is saturated, innocent, and pseudo-receptive. 
\end{theorem}
\begin{proof}
(Only if) If $\sigma : S \to A$ is a strong strategy then by definition it is strongly equivalent to $\cc_A \comp \sigma$, which is saturated, innocent, and strongly receptive (i.e. receptive and pseudo-receptive) by Lemma~9 of \cite{lics14}. The conclusion follows from Lemma~\ref{lem:stability}. 
(If) Assume $\sigma$ is innocent, pseudo-receptive and saturated. Since
$\sigma$ is saturated, $\sigma$ and $\Sat(\sigma)$ are strongly
equivalent, and therefore $\Sat(\sigma)$ is also innocent and
pseudo-receptive. Hence, by Theorem~\ref{thm:collapse},
$\collapse(\Sat(\sigma))$ and $\Sat(\sigma)$ are strongly equivalent,
and so $\sigma \strongequiv \collapse(\Sat(\sigma))$. Since strong
strategies are closed under strong equivalence, it now suffices to show
that $\collapse(\Sat(\sigma))$ is a strong strategy. The argument is as
follows: since $\sigma$ is pseudo-receptive, $\Sat(\sigma)$ satisfies
the existence part of receptivity by
Lemma~\ref{lem:existence-part-for-free}, and therefore by
Corollary~\ref{cor:collapse} the pre-strategy $\collapse(\Sat(\sigma))$
is innocent, pseudo-receptive, and receptive; so by
Theorem~\ref{thm:old-characterization} it is a strong strategy.
\end{proof}

\subsection{A characterization of weak strategies}
\label{subsec:weak-characterization}

\begin{theorem}
A pre-strategy $\sigma : S \to A$ is a weak strategy if and only if it is innocent and pseudo-receptive.
\end{theorem}
\begin{proof}
(Only if) If $\sigma : S \to A$ is a weak strategy then by definition it is weakly equivalent to $\cc_A \comp \sigma$, which is saturated, innocent, and strongly receptive (i.e. receptive and pseudo-receptive) by Lemma~9 of \cite{lics14}. Innocence
and pseudo-receptivity can be transported to $\sigma$, by Lemma~\ref{lem:stability}. (If) Recall $\sigma$ is weakly equivalent to $\sat{\sigma}$. The latter is saturated, innocent, and pseudo-receptive, and therefore a strong strategy by Theorem~\ref{thm:strong-characterization}. Since $\sigma$ is weakly equivalent to a strong strategy it is itself a weak strategy. 
\end{proof}

 It remains to show that innocence and pseudo-receptivity together correspond to the Street fibration condition. The proof involves reasoning about pseudo-lifts for maps in $\ESS$. 
\begin{theorem}[Strategies as Street fibrations]
\label{thm:characterization}
A pre-strategy $\sigma : S \to A$ is a weak strategy if and only if $\confsym{\sigma} : \confsym{S} \to \confsym{A}$ is a Street fibration. 
\end{theorem}

\begin{proof}
(If) Pseudo-receptivity is immediate because $\confsym{\sigma}$ is a Street fibration. We check innocence. Supposing $y \posincl \sigma x$, by assumption it has a pseudo-lift, i.e. a morphism
$z \to x$ in $\confsym{S}$ together with a symmetry $\sigma z \xiso{\theta} y$ such that the morphism $z \to x$ maps to $\sigma z \xiso{\theta} y \posincl \sigma x$. It is clear that in this case the $\negrev$ part of the morphism $z \to x$ must be trivial, so it is of the form $z \xiso{\phi} u \posincl x$ for a unique $\phi$ and $u$. Thus
$\sigma z \xiso{\sigma \phi} \sigma u \posincl \sigma x$. By uniqueness of factorization of morphisms of the Scott category,
we must therefore have $\sigma u = y$ and $\sigma (u \posincl x)= (y \posincl \sigma x)$.

(Only if) We must show that for every $x \in \confsym{S}$, every morphism
\begin{equation}
\label{eq:proof}
z \negrev t \toiso u \posincl \sigma x
\end{equation}
has an essentially unique pseudo-lift.

\emph{Existence.} Using the bisimulation property of symmetry, there exists $w$ such that the morphism in \eqref{eq:proof} can be written as a composite
\[
z \toiso w \negrev u \posincl \sigma x.
\]
Combining our two assumptions on $\sigma$ (pseudo-receptivity and innocence), we can find
$t$ and $v$ in $\config{S}$ such that $t \negrev v \posincl x$, and a symmetry $\sigma t \toiso w$, such that
\[
\sigma (t \negrev v \posincl x) \quad = \quad \sigma t \toiso w \negrev u \posincl \sigma x.
\]
The morphism $t \negrev v \posincl x$, together with the symmetry
$\sigma t \toiso w \toiso z$,
is therefore a pseudo-lift of $z \toiso w \negrev u \posincl \sigma x$
and thus of the morphism in \eqref{eq:proof}.

\emph{Essential uniqueness.}
Suppose there exist two pseudo-lifts of the morphism  in \eqref{eq:proof}, i.e. two morphisms
 \begin{align}
 v' \negrev w_0' \toiso u_0' \posincl x  \qquad\qquad  \label{eq1} v \negrev w_0 \toiso u_0 \posincl x
 \end{align}
in $\confsym{S}$, respectively mapping to
\[
\sigma v' \toiso z \negrev t \toiso u \posincl \sigma x
\qquad
\text{and}
\qquad
\sigma v \toiso z \negrev t \toiso u \posincl \sigma x
\]
in $\confsym{A}$, for given symmetries $\sigma v' \toiso z$ and $\sigma v \toiso z$.

Observe that the sub-configurations $u_0, u_0' \hookrightarrow x$ must be the same, since by assumption there is a unique strict lift of $u \posincl \sigma x$, and thus consider the two morphisms $v' \negrev w_0' \toiso u_0$ and $v \negrev w_0 \toiso u_0$. Via the bisimulation property of symmetry, they can be rewritten in `symmetry-first' format, as the dashed composites in the diagram below ($w_1$ and $w_1'$ are not necessarily unique, but the choice does not matter):
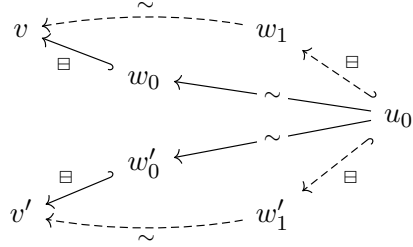
\[\begin{tikzcd}[row sep=0em]
	v && {w_1} & \\[0.3em]
	& {w_0} \\
	&&& {u_0} \\
	& {w_0'} \\
	{v'} && {w_1'}
	\arrow["\sim"', curve={height=6pt}, dashed, from=1-3, to=1-1]
	\arrow["\scriptscriptstyle\boxminus", hook, from=2-2, to=1-1]
	\arrow["\scriptscriptstyle\boxminus"', dashed, hook', from=3-4, to=1-3]
	\arrow["\sim"{description}, from=3-4, to=2-2]
	\arrow["\sim"{description}, from=3-4, to=4-2]
	\arrow["\scriptscriptstyle\boxminus", dashed, hook, from=3-4, to=5-3]
	\arrow["\scriptscriptstyle\boxminus"', hook', from=4-2, to=5-1]
	\arrow["\sim", curve={height=-6pt}, dashed, from=5-3, to=5-1]
\end{tikzcd}\]
In the same way, we refactor the morphism $z \negrev t \toiso u$ through some configuration $t'$: 
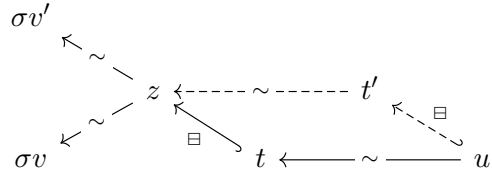
\[\begin{tikzcd}[row sep=1em]
	{\sigma v'} \\
	& z && {t'} \\
	\sigma v && t && u
	\arrow["\sim"{description}, from=2-2, to=1-1]
	\arrow["\sim"{description}, from=2-2, to=3-1]
	\arrow["\sim"{description}, dashed, from=2-4, to=2-2]
	\arrow["\scriptscriptstyle\boxminus", hook, from=3-3, to=2-2]
	\arrow["\scriptscriptstyle\boxminus"', dashed, hook', from=3-5, to=2-4]
	\arrow["\sim"{description}, from=3-5, to=3-3]
\end{tikzcd}\]
At this point we have constructed two pseudo-lifts of the morphism $t' \negrev u$, namely
$v \toiso w_1 \negrev u_0$ and $v' \toiso w_1' \negrev u_0$. But we have assumed that pseudo-lifts of negative reverse inclusions are essentially unique, so there is a unique isomorphism $v \toiso v'$ between the pseudo-lifts in \eqref{eq1}, as required.
\end{proof}

\section{Strategies as $\Setoid$-profunctors}
\label{sec:strategies-as-profunctors}

Our characterization of strategies as Street fibrations in setoids 
allows for an easy connection with setoid-valued profunctors. In this
section we formalize this as an oplax functor of tricategories
$\coll{-} : \Strat \to \SetoidProf$. The `oplaxness' means that
there is a non-invertible comparison map
$\coll{\tau \comp \sigma} \to \coll{\tau} \circ \coll{\sigma}$ because
viewing strategies as profunctors forgets important causal structure affecting
their composition.

The components of the functor $\coll{-}$ are straightforward to
describe. On objects, we define $\coll{A}$ to be the Scott category
$\confsym{A}$. Then we turn any strategy $\sigma : A \profto B$ into a setoid-valued profunctor $\coll{\sigma}_{A, B} : \confsym{A} \profto \confsym{B}$ by letting $\coll{\sigma}_{A, B}(x, y)$ be the essential fibre in $S$ over the configuration $x \parallel y \in \confsym{A^\perp \parallel B}$. This is always a setoid (Lemma~\ref{lem:ess-setoid-fibres}).

To define this more formally, first note the following property of the Scott category construction:
\begin{lemma}
For games $A$ and $B$, there is an isomorphism of categories $\confsym{A^\perp \parallel B} \cong \confsym{B} \times \confsym{A}^\op$.
\end{lemma}
From this we derive an isomorphism 
$ \SetoidSFib(\confsym{A^\perp \parallel B}) \to \SetoidSFib(\confsym{B} \times \confsym{A}^\op)$
 of 2-categories, whose action on a Street fibration is defined by post-composition. Additionally, since every strategy is a Street fibration, there is an embedding  $\Strat[A, B]
  \hookrightarrow \SetoidSFib(\confsym{A^\perp \parallel
      B})$ and recall  from \S\ref{sec:fibrations} the 2-functor      $\essfibres:
    \SetoidSFib(\confsym{B} \times \confsym{A}^\op) \to
    [\confsym{B}^\op \times \confsym{A}, \Setoid]$.
    We now assemble all this data.
\begin{theorem}
\label{thm:oplax-functor}
There is an oplax functor of tricategories $\coll{-} : \Strat \to \SetoidProf$. 
\end{theorem}
\begin{proof}
We only give the key steps of the construction.  The action on $\coll{-}$ on hom-2-categories is given by a family of 2-functors $\coll{-}_{A, B}$, for games $A$ and $B$, defined by
\[ 
  \Strat[A, B] \hookrightarrow \SetoidSFib(\confsym{A^\perp \parallel
      B})
      \to \SetoidSFib(\confsym{B} \times \confsym{A}^\op)
      \to \SetoidProf[\coll{A}, \coll{B}].
\]
For a functor of tricategories we need transformations relating identities and composition in each model. For identities, Lemma~\ref{lem:conf-copycat} characterizes the configurations of copycat in terms of the twisted arrow category, which corresponds to the identity profunctor (in particular our oplax functor is \emph{normal}). For composition, we use that every configuration of a composite strategy $T \comp S$ can be projected to a pair of `matching' configurations from $S$ and $T$. This is extended in a straightforward way to account for essential fibres and symmetries. 

Finally we verify that the axioms of pseudo-functors between bicategories hold, up to symmetry. This is done by repeated use of the universal property of the composition of profunctors; see \cite{pierre-hugo-lics23} for a similar bicategorical result. There are no further axioms to verify on the 3-cells because the structure is degenerate and all necessary equations hold automatically.
\end{proof}

\section{Thin games}
\label{sec:thin-games}

\newcommand{\TCG}{\mathbf{Thin}}
In this section we review an established approach to symmetry in concurrent games based on `thin' concurrent games and strategies \cite{cg2,lics15}. The theory of thin games, primarily developed by Clairambault and collaborators, has been very successful in applications of concurrent games to semantics \cite{clairambault2018fully,POPL20}. 

We show that $\Strat$ embeds the bicategory $\TCG$ of thin concurrent games. This result can be seen as challenging the common view that there are two separate approaches to symmetry in game semantics: the thin approach (e.g.~\cite{mellies2003asynchronous,cg2,paquet2023bi}) and the saturated approach (e.g.~\cite{lics14,mellies2019template}). We show that the model of weak strategies as Street fibrations provides a general mathematical universe in which the two approaches coexist and are formally connected. 
We only give the key definitions, but see \cite{cg2} for a full account.

\newcommand{\possym}[1]{\sym{#1}_+}
\newcommand{\negsym}[1]{\sym{#1}_-}
\begin{definition}
A \emph{thin game} is an event structure with symmetry $A = (A, \sym A)$ equipped with two sub-families of $\sym A$, denoted $\possym A$ and $\negsym A$, both satisfying the conditions of event structures with symmetry (Def.~\ref{def:ess}) and subject to the following axioms:
\begin{itemize}
\item If $\theta \in \possym{A}[x, y] \cap \negsym{A}[x, y]$, then $x = y$ and $\theta = \id_{x}$.
\item If $\theta \in \possym{A}[x, y]$, if 
$x \posincl x'$ and $\theta' \in \sym{A}[x', y']$ with $\theta'|_{x} = \theta$ then $\theta' \in \possym{A}[x, y]$.
\item If $\theta \in \negsym{A}[x, y]$, if 
$x \posincl x'$ and $\theta' \in \sym{A}[x', y']$ with $\theta'|_{x} = \theta$ then $\theta' \in \negsym{A}[x, y]$.
\end{itemize}
A \emph{thin strategy} on a thin game $(A, \possym A, \negsym A)$
is a pre-strategy $\sigma : S \to A$ (in the sense of
\S\ref{subsec:pre-strategies}, \ie\ $S = (S, \sym S)$ is an ordinary
event structure with symmetry) which is innocent and strongly
receptive, and additionally satisfies the following condition: for all $\theta \in \sym{S}$, if $\sigma \theta \in \possym{A}$ then $\theta = \id_x$ for some $x\in \config{S}$. 
\end{definition}

Parallel composition extends to thin games componentwise, and the
operation $(-)^\perp$ acts on thin games by additionally reverting the
roles of the positive and negative symmetries.

\newcommand{\CCt}{\CC^t}
\newcommand{\cct}{\cc^t}
\newcommand{\compt}{\comp^t}

Thin games and strategies are very constrained. The main advantage of these stronger axioms is that they enable simpler notions of identities and composition, with the associativity and unit laws holding up to isomorphism rather than equivalence. 

Formally, the identity on a thin game $A$ is a `thin copycat' strategy $\cct_A : \CCt_A \to A^\perp \parallel A$. It is (weakly) equivalent to the copycat $\CC_A$ from \S\ref{subsec:copycat}, but the latter is not a thin strategy in general. (The strategy $\CCt_A$ is in fact much simpler than $\CC_A$, but sadly not well-defined when $A$ is an arbitrary game, see \cite[A.4]{cg2}.) Thin strategies can be composed using a composition operation $\compt$ defined in the same way as $\comp$ but using 1-categorical pullbacks rather than pseudo-pullbacks. (The operation $\compt$ is simpler than $\comp$ but not well-defined for arbitrary non-thin strategies, as $\ESS$ does not have enough pullbacks.)

Altogether one gets a 2-dimensional model:
\begin{proposition}[\cite{hugo-thesis,clairambault:tel-04523273}]
There is a bicategory $\TCG$ of thin games, thin strategies between them as morphisms, and (weak) isomorphisms of strategies as 2-cells. 
\end{proposition}

\begin{theorem}
There is a pseudo-functor of tricategories $J : \TCG \to \Strat$, where the bicategory $\TCG$ is regarded as having only identity 3-cells.
\end{theorem}
\begin{proof} We give only a sketch of the construction. For a thin game $\calA = (A, \possym A, \negsym A)$ we define the plain game $J\calA = A$. A thin strategy is in particular a strategy in the sense of $\Strat$ and indeed this extends to an embedding $\TCG[\calA, \calB] \to \Strat[J\calA, J\calB]$. 

We must then relate identities and compositions. We have already stated that $\CC_A$ and $\CCt_A$ are equivalent strategies.  To show that $\comp$ and $\compt$ are also equivalent, one direction uses the universal property of the pseudo-pullback and the other that the pullback of thin strategies is in fact a bipullback \cite{cg2}.
\end{proof}

The functor $J : \TCG \to \Strat$ appears to forget important structure, namely the positive and negative sub-families. But, as we show next, this structure is `property-like': if a game admits a thin game structure then that structure is essentially unique from the point of view of the bicategory $\TCG$. The result is new to this paper. 
\begin{proposition}
\label{prop:embedding}
Let $\calA = (A, \possym A, \negsym A)$ and $\calA' = (A, \possym A', \negsym A')$ be thin games with the same underlying event structure with symmetry and polarity $A$. Then there is an internal equivalence $\calA \simeq \calA'$ in the bicategory $\TCG$.
\end{proposition}
\begin{proof}
An equivalence in $\TCG$ is given by a pair of strategies that are each other's pseudo-inverse. The key observation is that the thin copycat strategy $\CCt_A \to A^\perp \parallel A$ can simultaneously serve as a strategy $\calA \profto \calA'$ and a strategy $\calA' \profto \calA$, in addition to its usual role as identity strategy on $\calA$ and $\calA'$. That this forms an equivalence follows quickly from the fact that copycat is an identity strategy. The only minor difficulty is in verifying the thinness axiom for copycat in each case. This follows from the observation that any bijection in the intersection of $\possym{A}$ and $\negsym{A}'$ (similarly $\negsym{A}$ and $\possym{A}'$) must be an identity bijection. This can be shown by induction on the size of the bijection. We omit the details.
\end{proof}

\begin{remark}
The pseudo-functor $J$ is identity on morphisms and higher-cells, and Proposition~\ref{prop:embedding} makes it `essentially injective on objects' in a bicategorical sense. It would seem that $J$ can thus be regarded as a higher-categorical embedding, in a sense that is consistent with existing 1-categorical notions of property-like structure and embeddings \cite{kelly1997property,Carboni1994modulated}. We content ourselves with this informal observation, as more work is required for a formal statement: it is still unclear whether $J$ reflects equivalences in an appropriate sense and higher-categorical generalizations can be subtle. 
\end{remark}

\section{Conclusion}
We have presented a very general model of concurrent games on event structures with symmetry, in which strategies satisfy the laws of composition up to a weak notion of equivalence. This model brings together several existing approaches to symmetry \cite{lics14,lics15}. 
 
 Importantly, we have characterized weak strategies as Street fibrations. This gives them a canonical status and provides a mathematically straightforward connection with $\Setoid$-valued profunctors. This reinforces our view that weak notions of equivalences must be taken seriously in models of this kind. 

 An interesting problem is that of characterizing strong strategies in a
 similar style. It is known that all strong strategies define
 Grothendieck fibrations, but Example~\ref{ex:isofib-not-saturated}
 shows that the converse fails, so a characterization would necessarily rely on a stronger
 condition. The problem remains open, at least without restricting the symmetry on strategies.
 
 This work brings us closer to a formalization of concurrent games in a type theory (or proof assistant) that supports and encourages reasoning up to homotopy. We will pursue this in further work.

 \bibliographystyle{entics}
 \bibliography{biblio}

\appendix

\section{Stable families, event structures, and symmetry}
\label{app:stable-families}

\paragraph*{Stable families.}

Some proofs in Appendix~\ref{app:collapse} rely on stable families \cite{evstrs}, the key results on which are presented here. 

\begin{definition}
A \emph{stable family} is a nonempty family $\F$ of finite sets which is:
\begin{itemize}
  \item \emph{Complete:} if $Z \fsubseteq \F$ is compatible, meaning that there
    exists $x \in \F$ with $\bigcup Z \subseteq x$, then $\bigcup Z \in \F$;
  \item \emph{Stable:} if $Z \fsubseteq \F$ is compatible and nonempty, then
    $\bigcap Z \in \F$;
  \item \emph{Coincidence-free:} for all $x \in \F$ and all $e, e' \in x$ with
    $e \neq e'$, there exists $x_0 \in \F$ with $x_0 \subseteq x$  satisfying $e \in x_0 \iff e' \notin x_0$. 
\end{itemize}
\end{definition}

We call elements of $\F$ its {\em configurations} and $\bigcup\F$ its {\em events}. %

\begin{definition}
A {\em (total) map} $f:\F\to \G$  between stable families $\F$ and $\G$ is a %
function $f$ from the events of $\F$ to those of $\G$ such that 
for all   $x\in\F$
its direct image $f x\in\G$  and 
if  $e, e' \in x$ and $f(e) =f(e')$ then  %
$e=e'$. 
\end{definition}

\paragraph*{Stable families and event structures.}

The finite configurations of an event structure form a stable family, and the definition of maps between stable families ensures that the category of stable families embeds that of event structures. 

It turns out that the embedding functor has a right adjoint $\Pr$, thus giving a coreflection (an adjunction whose unit is an isomorphism). We recall, without proof, the construction of this adjoint. 

Let $x$ be a configuration of a stable family $\Fam$. 
Define the {\em prime} configuration of $e$ in $x$ by
$$[ e ]_x \eqdef \bigcap\ \{ y\in \Fam \mid  e\in y  \ \& \  y\subseteq x \}.
$$

\begin{theorem}\label{thm:Pr}
For $\Fam$ a stable family,   
$\Pr(\Fam) \eqdef 
(P, \Con, \leq)$ is an event structure where
 
\begin{itemize}
\item $P= \set{[ e ]_x }{e\in x\ \&\ x\in \Fam}$;
\item $Z \in \Con$ if and only if  $Z \subseteq P$ and $\hbox{$\bigcup$} Z \in\Fam$; and 
\item $p\leq p'$ if and only if $p, p'\in P$ and $p\subseteq p'.$
\end{itemize}
There is an order-isomorphism $\theta: (\config{\Pr(\Fam)},\subseteq) \iso (\Fam, \subseteq)$
where $\theta(y)\eqdef\max\, y = 
\bigcup y$ for $y\in \config{\Pr({\F})}$%
; its  mutual inverse is $\varphi$ where $\varphi(x) = \set{[e]_x}{e\in x}$ for  $x\in\Fam$.
\end{theorem}
By  coincidence-freeness, the  function 
  $\max:\config{\Pr({\F})}\to \F$ which takes  
a prime configuration $[e]_x$ to %
$e$ is well-defined;   it is  the 
counit  of the adjunction~\cite{icalp82,evstrs}.

\paragraph*{Symmetry as spans of open maps.}

In the main body of the paper we have presented symmetry on an event
structure $E$ as a family of isomorphisms $\sym E[x,y]$ for
configurations $x, y \in \config{E}$. But in fact (as we alluded to in
Remark~\ref{rem:spans}) the structure of a symmetry on $E$ can
equivalently be presented as another event structure $\sym E$,
together with a span of jointly monic open maps
\[
E \xleftarrow{\ l\ } \sym E \xrightarrow{\ r\ } E
\]
forming an internal equivalence relation in the category $\ES$. Each
configuration $w \in \config{\sym E}$ corresponds to an element
$\sym E[lw, rw]$ in the above presentation.

In the other direction, one constructs the event structure $\sym E$ by
observing that it forms a stable family (regarding each bijection as a
set of pairs) and applying the functor $\Pr$.

The above equivalence extends to morphisms of event structures with
symmetry. Indeed maps $f : (E, \sym E) \to (D, \sym D)$ of event
structures with symmetry (in the sense of \S\ref{sec:setoids}) are
equivalently presented as pairs of maps $f : E \to D$ and
$\sym f : \sym E \to \sym D$ of ordinary event structures, appropriately
commuting with the spans of open maps. Both perspectives are helpful and
we use the notation $\sym f$ in proofs of \S\ref{app:collapse}.

\section{The collapse of a pre-strategy}
\label{app:collapse}

\subsection{Preliminary terminology and notation}

Recall that a pre-strategy $\sig:S\to A$ is \emph{$\boxminus$-innocent} if whenever $s'\imc s$ with $\pol(s)=\boxminus$ in $S$ then $\sigma(s') \imc\sigma(s)$ in $A$. This condition is implied by pseudo-receptivity (Lemma~\ref{lem:minusinnocence}), but the results below typically only rely on $\boxminus$-innocence. 

Next, we say that a pre-strategy $\sigma:S\to A$ is {\em collapsed} if, for every $x\negincl x_1, x_2$ in $\config S$, $\sigma x_1 = \sigma x_2$ implies $x_1 = x_2$.  Provided a pre-strategy $\sigma:S\to A$  is 
pseudo-receptive we will show how to construct an equivalent collapsed pre-strategy $\collapse(\sig)$, the objective of this section.

For a configuration $x$ of any event structure with polarity $E$ we
write $x^{\bplus}$ and $x^{\bminus}$ for the subsets of positive and negative events,
respectively (these are not generally configurations themselves).
Similarly, let $E^{\bminus}$ and $E^{\bplus}$ denote the corresponding subsets of
events in $E$. And, when $E$ has symmetry and $\theta : x \toiso y$ is a
symmetry bijection, we denote by $\theta^{\bplus} : x^{\bplus} \iso y^{\bplus}$ and
$\theta^{\bminus} : x^{\bminus} \iso y^{\bminus}$ the appropriate restricted bijections (recall $\theta$ must preserve polarity). 

Finally, for an event $e$ of an event structure $E$, we let $[e]$ denote
its down-closure $\set{e' \in E}{e' \leq e}$. This is the smallest
configuration of $E$ containing $e$. Similarly, we write $[Z]$ for the down-closure of any consistent subset $Z \subseteq E$, which is always a configuration.  

\subsection{The collapse construction}

\begin{definition}
 \label{def:collapse}{\rm Let $\sigma:S\to A$ be a  $\boxminus$-innocent pre-strategy with essentially unique $\negincl$-lifts.  (In particular these assumptions hold if $\sig$ is pseudo-receptive.)
W.l.o.g.~assume that $S^{\bplus}$ and $A^{\bminus}$ are disjoint. 

Taking
$$
{\cal S} \eqdef \set{x^{\bplus}\cup (\sigma x)^{\bminus}}{x\in \config S}\,.
$$
defines
a stable family; this only depends on the $\boxminus$-innocence of $\sig$ ---see Lemma~\ref{lem:collapse} below.  (Noting that $(\sigma x)^{\bminus} = \sigma (x^{\bminus})$ we shall simply write $\sigma x^{\bminus}$ from now on.)

It is easy to see   that for any $x\in \config S$ there is a bijection
\begin{equation}
\label{eq:collapse-bij}
x\iso x^{\bplus} \cup \sigma x^{\bminus}
\end{equation}
taking $s \in x$ to $s$ if $s$ is positive
and to $\sig(s)$ if $s$ is negative.  Equip $\cal S$ with a symmetry $\sym{\cal S}$, the family comprising all composite bijections
$$
x^{\bplus} \cup \sigma x^{\bminus} \iso x \iso^\theta_S y \iso y^{\bplus} \cup \sigma y^{\bminus} 
$$
for $x\iso^\theta_S y$ in the isomorphism family of $S$, using the bijections \eqref{eq:collapse-bij}. Using that $\sig$  is both $\boxminus$-innocent and %
pseudo-receptive, we can show $\sym{\cal S}$ is an isomorphism %
family---Lemma~\ref{lem:collapse} below.
 
Define 
$\collapse(S)$ to be 
the event structure $\Pr({\cal S})$ with the polarity of $p\in \Pr({\cal S})$ the same as the polarity of $\max (p)$.  Its symmetry 
$\sym{\collapse(S)}=\Pr(\sym {\cal S})$ where $\sym{\cal S}$ is the
stable family $\sym{\cal S}$ which comprises all composite bijections
$$
x^{\bplus}\cup \sigma x^{\bminus} \iso  x \iso^\theta  y \iso y^{\bplus}\cup\sigma y^{\bminus}\,,
$$
where $x \iso_S^\theta y$. %
The family $\sym{\cal S}$  is shown to be a stable family below in Lemma~\ref{lem:collapse}. %
  
The prime configurations of $\cal S$ are precisely the configurations $[s]^{\bplus} \cup \sig[s]^{\bminus}$ for $s\in S$   ---see Lemma~\ref{lem:collapse} below.  (By definition, the prime configurations are the events of $\collapse(S)$.) 
Define the function 
$\collapse(\sig): \collapse(S) \to A$ to take a prime configuration $[s]^{\bplus} \cup \sig[s]^{\bminus}$, where $s\in S$, to $\sig(s)$. 
}\end{definition}

\begin{definition}\label{def:collapse-g}
Under the assumptions on $\sig$ above, define a function
$g:S\to \collapse(S)$ as the map taking 
$s\in S$ to the prime configuration $[s]^{\bplus} \cup \sig[s]^{\bminus}$. 
\end{definition}

To justify the claims in the definition of $\collapse$ above we first need a proposition on  a consequence of  $\boxminus$-innocence. 
\begin{prop}\label{prop:wneginnoc}
Let $\sigma:S\to A$ be a $\boxminus$-innocent  pre-strategy. Suppose that $x_0\subseteq x$ in $\config S$ and that $
\sigma x_0 \negincl y \subseteq \sigma x$ in $\config A$.  Then there is a (necessarily unique) $z\in \config S$ such that $x_0\negincl z \subseteq x$ and $\sigma z =y$.  
\end{prop}
\begin{proof}
Assume  $x_0\subseteq x$ in $\config S$ and
$\sigma x_0 \negincl y \subseteq \sigma x$
in $\config A$.   
Define the set 
$
z\eqdef \set{s\in x}{\sig(s) \in y}.
$

Clearly $x_0\negincl z \subseteq x$ and $\sigma z = y$.  
It remains to show that $z$ is a configuration of $S$. 
As $z\subseteq x$ the set $z$ is clearly consistent.  
So for $z$ to be a configuration of $S$ it suffices to show it down-closed.  And for this it suffices to
show for all $s \in z\setdif x_0$ if $s'\imc_S s$ then $s' \in z$ ---this property already holds of all events in $x_0$ as it is a configuration.   
Suppose $s \in z\setdif x_0$.  Then $s$ is negative, so as $\sig$ is  $\boxminus$-innocent, $\sig(s')\imc_A \sig(s)$. As $\sig(s)\in y$, and $y$ is down-closed,  $\sig(s')\in y$ making $s'\in z$ ---and $z$ down-closed so a configuration.
\end{proof}

\begin{lemma}\label{lem:collapse}
Above, in Definition~\ref{def:collapse},  on the assumption that the pre-strategy $\sigma:S\to A$ is $\boxminus$-innocent, 
\begin{enumerate}
\item[(A)]
the family $\cal S$ is stable. 
\end{enumerate}
Moreover, on the further assumption that $\sig$ is pseudo-receptive,
\begin{enumerate}
\item[(B)]
 the family $\sym{\cal S}$ is an isomorphism family, alternatively described by
 $
\sym{\cal S} = \set{\theta^{\bplus} \cup \sym\sigma \theta^{\bminus}}{\theta \in \sym S}\,;
$
 \item[(C)]
  the prime configurations of $\cal S$ are precisely those of the form $[s]^{\bplus} \cup \sigma [s]^{\bminus}$, for $s\in S$; and 
 the functions $\collapse(\sig)$ and $g$ are maps 
 of event structures with polarity for which $\collapse(\sig) g = \sig\,,$
 such that 
 \item[(D)]
 if $z\negincl z_1, z_2$ and 
$\collapse(\sig) z_1 =  \collapse(\sig) z_1$ then $z_1=z_2$, for all $z, z_1, z_2\in \config{\collapse(S)}$; and 
 \item[(E)]
 $g x_1 = g x_2 \implies x_1\iso_S x_2$, for all $x_1, x_2\in \config S$. 
 \end{enumerate}
\end{lemma}
\begin{proof}

 (A). 
We show $\cal S$ satisfies the three axioms of stable families, namely\footnote{It suffices to show binary versions of completeness and stability. We use $\uparrow$ to denote binary compatibility: $u \uparrow v$ means that $u$ and $v$ admit an upper bound in the family.}: (i) if $u \uparrow v$  in ${\cal S}$ then $u\cup v\in {\cal S}$; (ii) if $u \uparrow v$ in ${\cal S}$ then $u\cap v\in {\cal S}$; and (iii) if $e_1, e_2$ are distinct events in $u \in {\cal S}$ then there is $v \in {\cal S}$ with $v\subseteq u$ and $(e_1\in v \iff e_2 \notin v).$
\begin{enumerate}
\item[(i)] Suppose that $x_1, x_2, x \in\config S$ satisfy $x_1^{\bplus} \cup \sigma x_1^{\bminus}  \subseteq x^{\bplus} \cup \sigma x^{\bminus}$ and $x_2^{\bplus} \cup \sigma x_2^{\bminus}  \subseteq x^{\bplus} \cup \sigma x^{\bminus}$ in ${\cal S}$. 
Then, $[x_1^{\bplus} \cup x_2^{\bplus}] \subseteq x$ in $\config S$ and $
\sigma [x_1^{\bplus} \cup x_2^{\bplus}] \negincl \sigma x_1 \cup \sigma x_2 \subseteq  \sigma x
$ in $\config A$. By Proposition~\ref{prop:wneginnoc}, there is $z\in\config S$ with $\sigma z = \sigma x_1 \cup \sigma x_2$ and $[x_1^{\bplus} \cup x_2^{\bplus}] \negincl z \subseteq x.$ Then, $z^{\bplus} \cup \sigma z^{\bminus}\in  {\cal S}$ with 
$
z^{\bplus} \cup \sigma z^{\bminus} = (x_1^{\bplus} \cup x_2^{\bplus}) \cup (\sigma x_1 \cup \sigma x_2)^{\bminus}$. 
\item[(ii)]
Again, suppose $x_1, x_2, x \in\config S$ satisfy $x_1^{\bplus} \cup \sigma x_1^{\bminus}  \subseteq x^{\bplus} \cup \sigma x^{\bminus}$ and $x_2^{\bplus} \cup \sigma x_2^{\bminus}  \subseteq x^{\bplus} \cup \sigma x^{\bminus}$ in ${\cal S}$. We first claim there is $y\in\config A$, necessarily unique, such that $\sigma ([x_1^{\bplus}] \cap [x_2^{\bplus}]) \negincl y  \posincl \sigma x_1  \cap \sigma x_2.
$
To see this suppose $a\in  \sigma x_1  \cap \sigma x_2$ has negative polarity.  Then if 
$a' < a$ with $a'$ positive there is $s\in [x_1^{\bplus}] \cap [x_2^{\bplus}]$ for which $\sig(s) =a'$.  Thus, $a'\in \sigma ([x_1^{\bplus}] \cap [x_2^{\bplus}])$.
Hence we achieve the claim by 
 taking
$y\eqdef \sigma ([x_1^{\bplus}] \cap [x_2^{\bplus}]) \cup \set{a\in \sigma x_1  \cap \sigma x_2}{ a\hbox{ is negative} }.$ 

Now $\sigma ([x_1^{\bplus}] \cap [x_2^{\bplus}]) \negincl y  \subseteq  \sigma x.$ By  Proposition~\ref{prop:wneginnoc}, there is $z\in\config S$ with $\sigma z = y$ and
$
[x_1^{\bplus}] \cap [x_2^{\bplus}] \negincl z \subseteq x.
$
We have
$
z^{\bplus} \cup \sigma z^{\bminus} = (x_1^{\bplus} \cap x_2^{\bplus})\cup ( \sigma x_1^{\bminus}  \cap \sigma x_2^{\bminus}),
$ as required. 

\item[(iii)]
Suppose $e_1, e_2 \in x^{\bplus} \cup \sigma x^{\bminus}$  are distinct in the configuration of ${\cal S}$ obtained from $x\in \config S$. 
As remarked in Definition~\ref{def:collapse}, there is a bijection $x\iso x^{\bplus} \cup \sigma x^{\bminus}$. So 
in all cases there is a subconfiguration $v$ of $x$ such that $v^{\bplus}\cup \sigma v^{\bminus}$ separates them. For instance, assuming $e_1=s_1\in x^{\bplus}$ and $e_2=\sig(s_2)\in \sigma x^{\bminus}$ we can take $v$ to be a subconfiguration of $x$  containing one but not the other of $s_1$ and $s_2$. 
\end{enumerate}
We have shown $\cal S$ is a stable family and so determines an event structure $\collapse({\cal S}) \eqdef \Pr({\cal S})$. 

Now, further assuming that $\sig$ %
is pseudo-receptive, we will show properties (B)--(E). It will be helpful to note that, for $x_1, x_2\in \config S$, if $x_1^{\bplus} \cup \sigma x_1^{\bminus} = x_2^{\bplus} \cup \sigma x_2^{\bminus}$ then $x_1^{\bplus}  =  x_2^{\bplus}$ and $\sigma x_1 =\sigma x_2$ and we claim that moreover there exists a symmetry $x_1 \iso_S^\theta x_2$ such that $\id_{[x_1^{\bplus}]} \subseteq \theta$ and $\sym\sig\theta =\id_{\sigma x_1}.$ Indeed, as $x_1^{\bplus} = x_2^{\bplus}$, their down-closures, the configurations $[x_1^{\bplus}]$ and $[x_2^{\bplus}]$ are equal and hence $\sigma [x_1^{\bplus}] = \sig[x_2^{\bplus}]\negincl \sigma x_1 = \sigma x_2.$
 As $\sig$  has essentially unique $\negincl$-lifts, we obtain $x_1 \iso_S^\theta x_2$ as above.  

(B). 
By definition, $\sym{\cal S}$ comprises all bijections $\theta': x^{\bplus}\cup\sigma x^{\bminus} \iso x \iso_S^\theta y \iso y^{\bplus}\cup\sigma y^{\bminus}$
as $\theta$ ranges over the isomorphism family of $S$. 
It can be checked that $\theta' = \theta^{\bplus} \cup \sym\sigma \theta^{\bminus}$, so that $\sym{\cal S} = \set{\theta^{\bplus} \cup \sym\sigma \theta^{\bminus}}{\exists x, y.\ x\iso_S^\theta y}.$

We
prove that
  $\sym{\cal S}$ inherits the properties required of an isomorphism family from the isomorphism family of $S$.
 It is easy to see reflexivity and symmetry of $\sym{\cal S}$.  To check transitivity, suppose given $\theta': x^{\bplus}\cup\sigma x^{\bminus} \iso x\iso_S^\theta y \iso y^{\bplus}\cup\sigma y^{\bminus}$ and $\phi': z^{\bplus}\cup\sigma z^{\bminus} \iso y\iso_S^\phi w \iso w^{\bplus}\cup\sigma w^{\bminus}$
such that $y^{\bplus}\cup\sigma y^{\bminus}  = z^{\bplus}\cup\sigma z^{\bminus}.
$
 By the paragraph below (A) above, there is $y \iso_S^\psi z$ which restricts to the identity on $y^{\bplus} = z^{\bplus}$ and induces the identity between $\sigma y^{\bminus}$ and $\sigma z^{\bminus}$.  The composite bijection
 $x^{\bplus}\cup\sigma x^{\bminus} \iso x\iso_S^{\phi\psi\theta} w \iso w^{\bplus}\cup\sigma w^{\bminus}$ provides the composition $\phi'\theta'$. 
  
  We turn to the extension and restriction properties required of an isomorphism family. Here it will be useful to observe  a
 necessary and sufficient condition for one bijection to extend another in $\sym{\cal S}$. 

\emph{Observation:}  Let $\theta': x_1^{\bplus}\cup\sigma x_1^{\bminus} \iso x_1\iso_S^\theta y_1 \iso y_1^{\bplus}\cup\sigma y_1^{\bminus}$ and $\phi': x^{\bplus}\cup\sigma x^{\bminus} \iso x\iso_S^\phi y \iso y^{\bplus}\cup\sigma y^{\bminus}$
be bijections in $\sym{\cal S}$.  The bijection $\theta'$ extends the bijection $\phi'$, \ie~$\phi' \subseteq \theta'$,  iff both the following commute,
$$
  \begin{tikzcd}[row sep=1.8em, column sep=1.8em]
	{x_1} & {y_1} \\
	{[x^{\bplus}]} & {[y^{\bplus}]}
	\arrow["{\theta}"', from=1-1, to=1-2]
	\arrow[hook, from=2-1, to=1-1]
	\arrow["{\phi}"', from=2-1, to=2-2]
	\arrow[hook, from=2-2, to=1-2]
  \end{tikzcd}
  \hbox{\quad  and \quad }
  \begin{tikzcd}[row sep=1.8em, column sep=1.8em]
	{\sigma x_1} & {\sigma y_1} \\
	{\sigma x} & {\sigma y\,.}
	\arrow["{\sym\sigma \theta}"', from=1-1, to=1-2]
	\arrow[hook, from=2-1, to=1-1]
	\arrow["{\sym\sigma \phi}"', from=2-1, to=2-2]
	\arrow[hook, from=2-2, to=1-2]
  \end{tikzcd}
  $$
  To prove the observation, it suffices to consider the two cases: events in $x^{\bplus}$ and events in $\sigma x^{\bminus}$. We omit the details.
  
Resuming the proof that $\sym{\cal S}$ is an isomorphism family, suppose 
  $x^{\bplus}\cup\sigma x^{\bminus} \iso x\iso_S^\theta y \iso y^{\bplus}\cup\sigma y^{\bminus}$ and $v^{\bplus}\cup\sigma v^{\bminus} \subseteq x^{\bplus}\cup\sigma x^{\bminus}$ for configurations $x, y,v \in\config S$. 
  As $[v^{\bplus}] \negincl v$, $\sigma [v^{\bplus}] \negincl \sigma v\subseteq \sigma x.$
  Clearly $ [v^{\bplus}] \subseteq x$. So, by Proposition~\ref{prop:wneginnoc}, there is a unique $x_1\in\config S$ such that $[v^{\bplus}] \negincl  x_1 \subseteq x$ and $\sigma x_1 =\sigma v.$ Therefore $x_1^{\bplus}= v^{\bplus}$ which with  $\sigma x_1 =\sigma v$ implies $x_1^{\bplus} \cup\sigma x_1^{\bminus} = v^{\bplus}\cup\sigma v^{\bminus}.$

  Now we have
  $x\iso_S y$ and $x_1\subseteq x$, and from the isomorphism family of $S$, we obtain (a unique) $y_1\in\config S$ with $x_1\iso_S y_1$ a restriction of $x\iso_S y$, \ie
\begin{equation}
  \label{eq:diagram}
  \begin{tikzcd}[row sep=1.8em, column sep=1.8em]
	x & y \\
	{x_1} & {y_1}
	\arrow["{\sim}"'{outer sep=-1.2ex}, from=1-1, to=1-2]
	\arrow[hook, from=2-1, to=1-1]
	\arrow["{\sim}"'{outer sep=-1.2ex}, from=2-1, to=2-2]
	\arrow[hook, from=2-2, to=1-2]
\end{tikzcd}
\end{equation}
commutes. 
  As required, we obtain a sub-configuration $y_1^{\bplus} \cup\sigma y_1^{\bminus} $ of $y^{\bplus} \cup\sigma y^{\bminus}$ and a bijection
  $$
   v^{\bplus}\cup\sigma v^{\bminus} = x_1^{\bplus} \cup\sigma x_1^{\bminus} \iso x_1 \iso_S y_1 \iso y_1^{\bplus} \cup\sigma y_1^{\bminus} \,,
   $$
   in $\sym{\cal S}$ which is a restriction of the original bijection $x^{\bplus}\cup\sigma x^{\bminus} \iso y^{\bplus}\cup\sigma y^{\bminus}$.  To show this, the \emph{observation} says that diagrams $$
   \begin{tikzcd}[row sep=1.8em, column sep=1.8em]
	x & y \\
	{[v^{\bplus}]} & {[y_1^{\bplus}]}
	\arrow["{\sim}"'{outer sep=-1.2ex}, from=1-1, to=1-2]
	\arrow[hook, from=2-1, to=1-1]
	\arrow["{\sim}"'{outer sep=-1.2ex}, from=2-1, to=2-2]
	\arrow[hook, from=2-2, to=1-2]
   \end{tikzcd}
   \hbox{ \quad and \quad}
   \begin{tikzcd}[row sep=1.8em, column sep=1.8em]
	{\sigma x} & {\sigma y} \\
	{\sigma v} & {\sigma y_1}
	\arrow["{\sim}"'{outer sep=-1.2ex}, from=1-1, to=1-2]
	\arrow[hook, from=2-1, to=1-1]
	\arrow["{\sim}"'{outer sep=-1.2ex}, from=2-1, to=2-2]
	\arrow[hook, from=2-2, to=1-2]
   \end{tikzcd}
   $$
     must commute. But this follows from the commuting diagram \eqref{eq:diagram}, recalling that $v^{\bplus} = x_1^{\bplus}$ and $\sigma v = \sigma x_1$.

  For the remaining condition, suppose $x^{\bplus}\cup\sigma x^{\bminus} \iso x\iso_S^\theta y \iso y^{\bplus}\cup\sigma y^{\bminus}$ and $x^{\bplus}\cup\sigma x^{\bminus} \subseteq v^{\bplus}\cup\sigma v^{\bminus}$, 
  for configurations $x, y,v \in\config S$. This clearly entails
  $\sigma x\subseteq \sigma v$.  It is easy to check that the bijection  $x\iso_S^\theta y$ restricts to make the diagram
  $$
  \begin{tikzcd}[row sep=1.8em, column sep=1.8em]
	x & y \\
	{[x^{\bplus}]} & {[y^{\bplus}]}
	\arrow["{\theta}", from=1-1, to=1-2]
	\arrow["{\scriptscriptstyle\boxminus}", hook, from=2-1, to=1-1]
	\arrow["{\sim}"'{outer sep=-1.2ex}, from=2-1, to=2-2]
	\arrow["{\scriptscriptstyle\boxminus}"', hook, from=2-2, to=1-2]
  \end{tikzcd}
  $$
  commute.
  In particular,   $[x^{\bplus}]\iso_S^\theta  [y^{\bplus}]$ with  $[x^{\bplus}] \subseteq v$ so there is $w\in\config S$ and a commuting diagram
\begin{equation}
  \label{eq:diagram2}
  \begin{tikzcd}[row sep=1.8em, column sep=1.8em]
	v & w \\
	{[x^{\bplus}]} & {[y^{\bplus}]\,.}
	\arrow["{\sim}"'{outer sep=-1.2ex}, from=1-1, to=1-2]
	\arrow[hook, from=2-1, to=1-1]
	\arrow["{\sim}"'{outer sep=-1.2ex}, from=2-1, to=2-2]
	\arrow[hook, from=2-2, to=1-2]
  \end{tikzcd}
\end{equation}
As $\sigma x\subseteq \sigma v$ we have $\sigma [x^{\bplus}] \negincl \sigma x \subseteq \sigma v,$ and by the above a commuting diagram,
   $$
  \begin{tikzcd}[row sep=1.8em, column sep=1.8em]
	{\sigma v} & {\sigma w} \\
	{\sigma x} & {} \\
	{\sigma [x^{\bplus}]} & {\sigma [y^{\bplus}]\,.}
	\arrow["{\sim}"'{outer sep=-1.2ex}, from=1-1, to=1-2]
	\arrow[hook, from=2-1, to=1-1]
	\arrow["{\scriptscriptstyle\boxminus}"', hook, from=3-1, to=2-1]
	\arrow["{\sim}"'{outer sep=-1.2ex}, from=3-1, to=3-2]
	\arrow[hook, from=3-2, to=1-2]
  \end{tikzcd}
  $$
  As $\iso_A$ forms an isomorphism family, this determines a unique $z\in\config A$ for which
\begin{equation}
  \label{eq:diagram4}
  \begin{tikzcd}[row sep=1.8em, column sep=1.8em]
	{\sigma v} & {\sigma w} \\
	{\sigma x} & z \\
	{\sigma [x^{\bplus}]} & {\sigma [y^{\bplus}]\,.}
	\arrow["{\sim}"'{outer sep=-1.2ex}, from=1-1, to=1-2]
	\arrow[hook, from=2-1, to=1-1]
	\arrow["{\sim}"'{outer sep=-1.2ex}, from=2-1, to=2-2]
	\arrow[hook, from=2-2, to=1-2]
	\arrow["{\scriptscriptstyle\boxminus}"', hook, from=3-1, to=2-1]
	\arrow["{\sim}"'{outer sep=-1.2ex}, from=3-1, to=3-2]
	\arrow["{\scriptscriptstyle\boxminus}", hook, from=3-2, to=2-2]
  \end{tikzcd}
\end{equation}
 commutes. As $[y^{\bplus}] \subseteq w$, by Proposition~\ref{prop:wneginnoc}, there is a unique $w_1\in\config S$ such that $[y^{\bplus}] \negincl w_1 \subseteq w \quad \& \quad \sigma w_1 =z.$  This makes $ [y^{\bplus}] \negincl w_1$ an $\negincl$-lift of $ \sigma [y^{\bplus}]\negincl z$.
 We also have that
 $$
  \begin{tikzcd}[row sep=1.8em, column sep=1.8em]
	{\sigma x} & {\sigma y} \\
	{\sigma [x^{\bplus}]} & {\sig[y^{\bplus}]\,.}
	\arrow["{\sim}"'{outer sep=-1.2ex}, from=1-1, to=1-2]
	\arrow["{\scriptscriptstyle\boxminus}"', hook, from=2-1, to=1-1]
	\arrow["{\sim}"'{outer sep=-1.2ex}, from=2-1, to=2-2]
	\arrow["{\scriptscriptstyle\boxminus}", hook, from=2-2, to=1-2]
  \end{tikzcd}
  $$
commutes.  This provides another  (essential) lift $  [y^{\bplus}] \negincl y$, this time with $\sigma y \iso_A z$.  As $\sig$ is pseudo-receptive we obtain the commuting diagram
 $$
 \begin{tikzcd}[row sep=1.8em, column sep=1.8em]
	{w_1} & y \\
	{[y^{\bplus}]\,,}
	\arrow["{\sim}"'{outer sep=-1.2ex}, from=1-1, to=1-2]
	\arrow["{\scriptscriptstyle\boxminus}"', hook, from=2-1, to=1-1]
	\arrow["{\scriptscriptstyle\boxminus}"{description}, hook, from=2-1, to=1-2]
 \end{tikzcd}
  $$
 which, by the isomorphism family of $S$, we can extend to
\begin{equation}
  \label{eq:diagram5}
 \begin{tikzcd}[row sep=1.8em, column sep=1.8em]
	w & {y_1} \\
	{w_1} & y \\
	{[y^{\bplus}]\,,}
	\arrow["{\sim}"'{outer sep=-1.2ex}, from=1-1, to=1-2]
	\arrow[hook, from=2-1, to=1-1]
	\arrow["{\sim}"'{outer sep=-1.2ex}, from=2-1, to=2-2]
	\arrow[hook, from=2-2, to=1-2]
	\arrow["{\scriptscriptstyle\boxminus}"', hook, from=3-1, to=2-1]
	\arrow["{\scriptscriptstyle\boxminus}"{description}, hook, from=3-1, to=2-2]
 \end{tikzcd}
\end{equation} 
  for some $y_1\in\config S$. Because $y\subseteq y_1$ we get $y^{\bplus}\cup\sigma y^{\bminus} \subseteq y_1^{\bplus}\cup\sigma y_1^{\bminus}$. 
  Also, the bijection
  $$
  v^{\bplus}\cup\sigma v^{\bminus} \iso v \iso_S w \iso_S y_1 \iso y_1^{\bplus}\cup\sigma y_1^{\bminus}
  $$
  in $\sym{\cal S}$ 
    extends the original bijection $x^{\bplus}\cup\sigma x^{\bminus} \iso y^{\bplus}\cup\sigma y^{\bminus}$ 
  in $\sym{\cal S}$.  To see this, use the \emph{observation} above. Combine the commuting diagrams \eqref{eq:diagram2} and \eqref{eq:diagram5}, to obtain that
  $$
   \begin{tikzcd}[row sep=1.8em, column sep=1.8em]
	v & {y_1} \\
	{[x^{\bplus}]} & {[y^{\bplus}]}
	\arrow["{\sim}"'{outer sep=-1.2ex}, from=1-1, to=1-2]
	\arrow[hook, from=2-1, to=1-1]
	\arrow["{\sim}"'{outer sep=-1.2ex}, from=2-1, to=2-2]
	\arrow[hook, from=2-2, to=1-2]
   \end{tikzcd}$$
  commutes.  From \eqref{eq:diagram4}, noting $z=\sigma w_1$, combined with the application of $\sig$ to the top commuting square of \eqref{eq:diagram5} we obtain that 
   $$
   \begin{tikzcd}[row sep=1.8em, column sep=1.8em]
	{\sigma v} & {\sigma y_1} \\
	{\sigma x} & {\sigma y}
	\arrow["{\sim}"'{outer sep=-1.2ex}, from=1-1, to=1-2]
	\arrow[hook, from=2-1, to=1-1]
	\arrow["{\sim}"'{outer sep=-1.2ex}, from=2-1, to=2-2]
	\arrow[hook, from=2-2, to=1-2]
   \end{tikzcd}$$
  commutes. 
 \\

  (C). 
We show that the functions $\collapse(\sig)$ and $g$ are maps of event structures with polarity.

First, a characterization of the primes of $\cal S$, which constitute the events of $\collapse(S)$. 
Let $x\in\config S$. 
From $(1)$ above, independently of the choice of $x$, the sub-configurations of 
$x^{\bplus}\cup\sigma x^{\bminus} $ are exactly $y^{\bplus}\cup\sigma y^{\bminus}$,  
where $y$ is a sub-configuration of $x$. It follows that prime configurations of $\cal S$ always take the form 
$[s]^{\bplus} \cup \sigma [s]^{\bminus}$, for some $s\in S$. 

The function $\collapse(\sig): \collapse(S) \to A$ takes a prime $[s]^{\bplus} \cup \sigma [s]^{\bminus}$ to $\sig(s)$.  A typical configuration of $\collapse(S)$ is a set
$$
\set{p\in \Pr({\cal S})}{p \subseteq x^{\bplus}\cup\sigma x^{\bminus} }\,,
$$
for some $x\in\config S$.  Hence  $\collapse(\sig)$ takes it to $\sigma x$ in a locally injective way.  

It is easy to see that $g: S \to \collapse(S)$ is a map.  It takes a configuration $x\in\config S$ to the set of prime sub-configurations of $x^{\bplus}\cup\sigma x^{\bminus} $; it is locally injective because of the bijection $x\iso x^{\bplus}\cup\sigma x^{\bminus} $.
In addition,  
$$\collapse(\sig) g (s) = 
\collapse(\sig)([s]^{\bplus} \cup \sigma [s]^{\bminus}) =\sig(s)\,,$$
 for all $s\in S$.  Hence, $\collapse(\sig) \circ g = \sig$.

  (D). %
 Assume  $z, z_1, z_2\in \config{\collapse(S)}$ with $z\negincl z_1$ and $z\negincl z_2$.  Suppose  
$\collapse(\sig) z_1 =  \collapse(\sig) z_2$.  Then $\bigcup z_1 = x_1^{\bplus} \cup \sigma x_1^{\bminus}$ and $\bigcup z_2 = x_2^{\bplus} \cup \sigma x_2^{\bminus}$, 
for some $x_1, x_2\in \config S$. From $z\negincl z_1, z_2$, we see that $x_1^{\bplus} = x_2^{\bplus}$ and, from $\collapse(\sig) z_1 =  \collapse(\sig) z_2$, that $\sigma x_1 = \sigma x_2$.  Therefore  
 $z_1=z_2$. 
 
  (E).   
  Finally we show that
$g x_1 = g x_2 \implies x_1\iso_S x_2$, for all $x_1, x_2\in \config S$. Supposing $g x_1 = g x_2$ we also have 
$x_1^{\bplus}\cup\sigma x_1^{\bminus}  = \bigcup g x_1 = \bigcup g x_2 =  x_2^{\bplus}\cup\sigma x_2^{\bminus}$ which implies $x_1\iso_S x_2$, by (1) above.
\end{proof}

\begin{lemma}\label{lem:collapseequiv}
If $\sig$ is any pseudo-receptive (and therefore $\boxminus$-innocent) pre-strategy,  %
$g:\sigma \strongequiv \collapse(\sig)$. 
\end{lemma}
\begin{proof}
Recall the 
map $g:S\to \collapse(S)$ which 
takes $s\in S$ to the prime configuration $[s]^{\bplus} \cup \sig[s]^{\bminus}$. 

We produce a converse map
$f:\collapse(S) \to S$
via Zorn's Lemma.
By Zorn's Lemma there is a $\subseteq$-maximal map $f: S_0 \to S$ such that  $g f =\id_{S_0}$ and 
$$
\begin{tikzcd}[row sep=1.8em, column sep=1.8em]
	{S_0} & S \\
	& A
	\arrow["f", from=1-1, to=1-2]
	\arrow["{\sig_0}"', from=1-1, to=2-2]
	\arrow["\sig", from=1-2, to=2-2]
\end{tikzcd}
$$
commutes, with a rigid embedding $S_0 \hookrightarrow \collapse(S)$ s.t. $\sig_0$ is %
the restriction of $\collapse(\sig)$ to $S_0$, \ie~$\sig_0$ is the composition
$$
\sig_0: 
S_0 \hookrightarrow \collapse(S) \arr{\collapse(\sig)} A\,.
$$
(Clearly the empty event structure embeds rigidly in $\collapse(S)$ and %
satisfies the above %
constraints;  Zorn's Lemma ensures this can be extended maximally.)

We show that $\sig_0 =\collapse(\sig)$.  Suppose otherwise, that $S_0 \neq \collapse(S)$.  Then, there is $p\in \Pr({\cal S})$ such that $p\notin S_0$ yet $[p)%
\in\config{S_0}$.  By assumption, $g f [p) = [p)$.  

In  Lemma~\ref{lem:collapse}, we saw that a 
prime configuration $p$ of $\cal S$, \ie~an event of $\collapse(S)$, has the form 
$$p=[s]^{\bplus} \cup \sigma [s]^{\bminus}\,,$$
 for some $s\in S$.
 We have $g[s) = [p)$ 
 so 
$g[s) = gf [p)$.  By Lemma~\ref{lem:collapse}, it follows that  
 $ [s) \iso_S f[p)$. 
 Because $[s) \longcov s [s]$ we also have $f[p) \longcov{s_0} [s_0]$ with $[s]\iso_S [s_0]$ extending $ [s) \iso_S f[p)$.
 
 Extend $S_0$ to $S_0'=S_0\cup \setof p$, a larger substructure of $\collapse(S)$. The map $f$ extends to a map $f':S_0'  \to S$ which  acts so $f'(p) =s_0$, and the function $\sig_0$ to $\sig_0': S_0' \to A$ so $\sig_0'(p) = \sig(s_0)$.  
Clearly  we still have $\sigma f' = \sig_0'$.
With the new value $p$, we maintain $g f' =\id_{S_0'}$, as we have
 $$
 g f'(p) = g(s) = [s]^{\bplus}\cup \sigma [s]^{\bminus} = p\,.
 $$
 This contradicts the maximality of $\sig_0$, unless $S_0 =\collapse(S)$ as we require. 
 
 To conclude we show $f$ and $g$ form   an equivalence $$f:\collapse(\sig) \strongequiv \sigma:g\,.$$ Certainly $g f =\id_{\collapse(S)}$. 
Let $x\in \config S$. We have
$$
g(fg\, x) = gfg\, x = \id_{\collapse(S)} g\, x =  g\, x\,,
$$
whence 
$$
fg\, x \iso_S x\,,
$$
using Lemma~\ref{lem:collapse}.
 \end{proof}

\subsection{Universality of the collapse}

Although we have not used it directly, we note the following universal property of the collapse construction.

\begin{corollary}(universal characterization)\label{cor:absview}
Let $\sigma:S\to A$ and $\sig':S'\to A$ be pre-strategies where $\sig$ is  $\boxminus$-innocent and pseudo-receptive and $\sig'$ is collapsed. %
 Then for a map of pre-strategies $f:\sigma \to \sig'$ there is a unique  map of pre-strategies $h: \collapse(\sig) \to \sig'$ such that $f=h g$:
$$
\begin{tikzcd}[row sep=1.8em, column sep=2.2em]
	S & {\collapse(S)} & {S'} \\
	& A &
	\arrow["f", dotted, bend left=32, from=1-1, to=1-3]
	\arrow["g", from=1-1, to=1-2]
	\arrow["\sig"', from=1-1, to=2-2]
	\arrow["{\collapse(\sig)}"{description}, from=1-2, to=2-2]
	\arrow["h", dashed, from=1-2, to=1-3]
	\arrow["{\sig'}", dotted, bend left=18, from=1-3, to=2-2]
\end{tikzcd}
$$

Amongst pre-strategies in $A$ which are $\boxminus$-innocent and pseudo-receptive,
 $\collapse$ is left adjoint to the inclusion of the subcategory of collapsed pre-strategies in the category of all pre-strategies.  
 
 If $\sigma:S \to A$ is  $\boxminus$-innocent and pseudo-receptive then  $\sym{\collapse(S)} \iso {\collapse(\sym S)}$ and $\sym{\collapse(\sig)} \iso {\collapse(\sym \sig)}$. 
 
 \end{corollary}
\begin{proof} (Sketched)
To verify the universal characterization, the map $h$ must take an event $[s]^{\bplus}\cup \sigma [s]^{\bminus}$ of $\collapse(S)$, where $s\in S$, to $f(s)$. 
The claimed adjunction follows directly from the universality.  Via the adjunction $\collapse$ extends to a functor. 

An event of $\sym{\cal S} $ is a pair of events of $\cal S$.  Define 
$L: \sym{\cal S} \to {\cal S}$ and $R: \sym{\cal S} \to {\cal S}$ to be the left and right projections on events.  
That  $\sym{\cal S}$ satisfies the axioms of an isomorphism family expresses that $L,R: \sym{\cal S} \to {\cal S}$ forms a symmetry in the category of stable families:
$$
\begin{tikzcd}[row sep=1.8em, column sep=1.8em]
	& {\sym{\cal S}} & \\
	{\cal S} && {\cal S}
	\arrow["L"', from=1-2, to=2-1]
	\arrow["R", from=1-2, to=2-3]
\end{tikzcd}
$$
Applying the right adjoint $\Pr$ to the inclusion functor of (families of configurations of) event structures to stable families, we obtain the symmetry
$$
\begin{tikzcd}[row sep=1.8em, column sep=1.8em]
	& {\Pr(\sym{\cal S})} & \\
	{\collapse(S)} && {\collapse(S)}
	\arrow["l"', from=1-2, to=2-1]
	\arrow["r", from=1-2, to=2-3]
\end{tikzcd}
$$
where we have written $l\eqdef \Pr(L)$ and $r\eqdef \Pr(R)$.  
By definition, $\sym{\collapse(S)} = \Pr(\sym{\cal S})$. 

Recall by Lemma~\ref{lem:collapse},
$$
\sym{\cal S} = \set{\theta^{\bplus} \cup \sym\sigma \theta^{\bminus}}{\exists x, y.\ x\iso_S^\theta y}\,.
$$
There is an order-isomorphism between finite configurations of ${\sym S}$, with symmetry maps $l_S,r_S:\sym S\to S$, and the isomorphism family of $S$: under the isomorphism 
$$
\config{\sym S} \iso \set{\theta}{\exists x, y.\ x\iso_S^\theta y}
$$
 $z\in\config{\sym S} $ is taken to the bijection $l_S z\iso z\iso r_S z$. This isomorphism induces an order-isomorphism
 $$
 \sym{\cal S} \iso \set{z^{\bplus} \cup \sym\sigma z^{\bminus}}{z\in\config{\sym S}}\,.
 $$
 Whence 
 $$
 \begin{aligned}
\sym{\collapse(S)} \eqdef  \Pr(\sym{\cal S}) \iso  &\Pr(\set{z^{\bplus} \cup \sym\sigma z^{\bminus}}{z\in\config{\sym S}})\\
 & \eqqcolon \collapse(\sym S)\,.
\end{aligned}
 $$  
 The isomorphism respects the left and right symmetry maps:
 $$
 \begin{tikzcd}[row sep=1.8em, column sep=1.8em]
	& {\sym{\collapse(S)}} & \\
	{\collapse(S)} && {\collapse(S)} \\
	& {\collapse(\sym S)} &
	\arrow["l"', from=1-2, to=2-1]
	\arrow["r", from=1-2, to=2-3]
	\arrow["{\rotatebox[origin=c]{-90}{$\iso$}}"{description}, phantom, from=1-2, to=3-2]
	\arrow["{\collapse(l_S)}", from=3-2, to=2-1]
	\arrow["{\collapse(r_S)}"', from=3-2, to=2-3]
 \end{tikzcd}
 $$
 commutes.
   It follows that 
 $$
 \sym{\collapse(\sig)} \iso {\collapse(\sym \sig)}\,.
 $$
\end{proof}

\end{document}